\documentclass[letterpaper,11pt,prb,reprint,aps,superscriptaddress,showpacs,nofootinbib]{revtex4-2}

\usepackage[utf8]{inputenc}
\usepackage[colorlinks=true,linkcolor=blue,citecolor=blue,urlcolor=blue]{hyperref}
\usepackage{graphicx}
\usepackage{amsmath,amssymb}
\usepackage{dsfont}
\usepackage{xcolor}
\usepackage{amsthm}
\usepackage{float}

\usepackage{algorithm}
\usepackage{algpseudocode}

\algnewcommand\algorithmicforeach{\textbf{for each}}
\algdef{S}[FOR]{ForEach}[1]{\algorithmicforeach\ #1\ \algorithmicdo}

\usepackage{braket}

\newtheorem{lemma}{Lemma}
\newtheorem{definition}{Definition}
\newtheorem{theorem}{Theorem}

\theoremstyle{remark}

\def\red{red!70}

\def\green{green!70!black!80}

\begin{document}

\title{Fast magic state preparation by gauging higher-form transversal gates in parallel}
\author{Dominic J. Williamson}
\affiliation{School of Physics, University of Sydney, Sydney, NSW 2006, Australia}
\date{January 2026}

\begin{abstract}
\noindent
Magic states are a foundational resource for universal quantum computation. 
To survive in a realistic noisy environment, magic states must be prepared fault-tolerantly and protected by a quantum error-correcting code. 
The recent discovery of highly efficient quantum low-density parity-check codes, together with efficient logic gates, lays the groundwork for low-overhead fault-tolerant quantum computation. 
This motivates the search for fast and parallel protocols for logical magic state preparation to enable universal quantum computation. 
Here, we introduce a fast code surgery procedure that performs a fault-tolerant measurement of many transversal logic gates in parallel. 
This is achieved by performing a generalized gauging measurement on a quantum code that supports a higher-form transversal gate. 
The time overhead of our procedure is constant, and the qubit overhead is linear. 
The procedure inherits fault-tolerance properties from the base code and the structure of the higher-form transversal gate. 
When applied to codes that support higher-form Clifford gates our procedure achieves fast and fault-tolerant preparation of many magic states in parallel. 
This motivates the search for good quantum low-density parity-check codes that support higher-form Clifford gates. 
\end{abstract}

\maketitle

\section{Introduction}


Utility-scale fault-tolerant quantum computers promise to deliver exponentially faster solutions to important problems in cryptography, chemistry, and physics, while functioning in a realistic noisy environment. 
The possibility of universal fault-tolerant quantum computing is well established, however, the timeline to its realization depends heavily on the cost of implementing quantum logic fault tolerantly~\cite{yoder2025tour}. 
In recent years, highly efficient \textit{good} quantum codes have been discovered that require only low-density parity checks~\cite{panteleev2022asymptotically}, see also Refs.~\cite{breuckmann2021balanced,leverrier2022quantum,dinur2022locally}. 
Fault-tolerant quantum logic operations have been developed for such codes which preserve their efficiency and low-density parity check structure while enabling universal quantum computation~\cite{gottesman2013fault,nguyen2024quantum,tamiya2024polylog}. 

In the standard circuit model, universal quantum computation is achieved by complementing nonuniversal Clifford operations with non-Clifford gates such as $CCZ$ or $T$. 
These non-Clifford gates can be performed via gate teleportation using Clifford operations and specially prepared nonstabilizer \textit{magic states}~\cite{NielsenChuang}. 
This division of operations has the advantage that it respects the structure of stabilizer codes, the most well-understood class of quantum codes~\cite{gottesman1997stabilizer}. 

Pauli Based Computation (PBC) is an alternative model for universal quantum computation that is well-suited to good quantum low-density parity-check (qLDPC) codes.
In this model, universal quantum computation is performed via adaptive Pauli measurements on Pauli and magic states~\cite{bravyi2016trading}. 
A quantum circuit is mapped to PBC by compiling out all Clifford operations. 
This is a natural fit for good qLDPC codes where a priori there is no preferred cheap logical basis. 
There has been a surge of recent progress on performing efficient fault-tolerant PBC on good qLDPC codes with the ultimate goal of implementing arbitrary sets of commuting logical Pauli measurements in parallel with the lowest possible space and time overhead~\cite{cohen2022low,cross2024improved,xu2024fast,williamson2024low,ide2024fault,swaroop2024universal,cowtan2025parallel,he2025extractors,Xu2025Batched,zheng2025high,baspin2025fast,Cowtan2025Fast}. 
The other equally important aspect of PBC is the preparation of magic states. 
The central goal on this front is to fault-tolerantly prepare as many encoded magic states as possible, in parallel, with minimal space and time overhead. 

Prior work on fault-tolerant magic state preparation has focused on codes that support transversal non-Clifford gates for direct application~\cite{bombin2015gauge,bombin2016dimensional,Vasmer2019three,Bombin2007exact,Bombin20182D,Brown2020FTNC,breuckmann2024cups,golowich2024quantum,scruby2024quantum,Hsin2024,Zhu2025b,zhu2023gates,zhu2025topological,Gulshen2025,Golowich2025,He2025,golowich2025constant,Tan2025Single,Jacob2025,Menon2025,Kobayashi2025}
or magic state distillation~\cite{bravyi2005magic,meier2012magic,Bravyi2012magic,campbell2012magic,Jones2013multilevel,haah2017magic,krishna2018magic,Krishna2019towards,litinski2019magic,rodriguez2024experimental,wills2024constant,Nguyen2024}. 
Another line of work has focused on the possibility of fault-tolerantly measuring individual transversal Clifford gates for topological codes of fixed size~\cite{li2015magic,Chamberland2019faulttolerantmagic,Chamberland2020verylow,gidney2024magic,gupta2024encoding,Hirano2024zerolevel,Itogawa2025,Chen2025,Vaknin2025,Sahay2025,Claes2025}, or in a scalable family~\cite{davydova2025universal,Huang2025,Bauer2025}. 

In this work we introduce an efficient, fast and parallel method that is capable of fault-tolerantly preparing magic states at a constant rate via the measurement of higher-form logical Clifford gates in suitable qLDPC codes. 
To achieve this, we define higher-form transversal gates for qLDPC codes, generalizing the concept of higher-form symmetry from the topological code setting~\cite{kapustin2013higher,Gaiotto2015}. 
We develop a higher-form gauging measurement procedure that can be applied to codes with higher-form transversal gates. 
This extends recently developed fast surgery procedures beyond Pauli operators~\cite{baspin2025fast,Cowtan2025Fast}. 

We apply the higher-form gauging procedure to prepare many magic states in parallel by measuring many commuting transversal Clifford operators simultaneously. 
For Calderbank-Shor-Steane (CSS)~\cite{steane1996multiple,calderbank1996good} codes, a sufficient, but not necessary, condition for a higher-form Clifford gate is the existence of a transversal non-Clifford $CCZ$ or $T$ gate. 
Our approach presents an alternative path to efficient fault-tolerant magic state preparation with qLDPC codes that does not rely on transversal non-Clifford gates or magic state distillation, but instead requires Clifford gates with a higher-form structure. 
This raises the challenge of discovering efficient qLDPC codes with single-shot Pauli state preparation and 1-form Clifford gates for the purpose of magic state preparation.

\section{Background}

In this section we introduce the setting for our work. 
We focus on qLDPC codes defined on qubits, the extension to qudits is straightforward~\cite{breuckmann2021ldpc}. 
A family of qLDPC codes is defined to be the simultaneous $+1$ eigenspace of a collection of constant-weight Hermitian check operators, such that each qubit is only acted on by a constant number of checks. 
The \textit{support} of an operator is the set of qubits it acts on nontrivially, the \textit{weight} of the operator is the size of this set. 
A Pauli stabilizer code is one in which all checks are products of Pauli operators, which must commute with each other to define a nontrivial codespace. 
A CSS code is a stabilizer code where a generating set of checks are products of either Pauli-$X$ operators or Pauli-$Z$ operators~\cite{steane1996multiple,calderbank1996good}. 
The parameters of a quantum code are denoted $[[n,k,d]]$ where $n$ is the number of physical qubits, $k$ is the number of encoded logical qubits, and $d$ is the \textit{distance}, minimum weight of a nontrivial logical operator. 
A \textit{logical operator} is an operator that preserves the codespace. 

A \textit{chain complex} $C_\bullet$ over $\mathbb{F}_2$ is a sequence of linear \textit{boundary} maps between finite dimensional vector spaces over $\mathbb{F}_2$ which are viewed as Abelian groups under addition
\begin{align}
    \dots \mathop{\leftrightarrows}^{\partial_{0}}_{\delta_{0}} 
    C_0 \mathop{\leftrightarrows}^{\partial_{1}}_{\delta_{1}} 
    C_1 \mathop{\leftrightarrows}^{\partial_2}_{\delta_2} 
    C_2 \mathop{\leftrightarrows}^{\partial_3}_{\delta_3} \dots
\end{align}
satisfying the condition $\partial_i \partial_{i+1}=0$. 
Here, we identify each finite dimensional vector space $C_i$ with its dual space $\widehat{C}_i$ and hence the \textit{coboundary} maps are given by $\delta_i = \partial_i^T$, which satisfy ${\delta_{i+1} \delta_{i}=0}$. 
The $i$-th \textit{homology group} is defined to be $H_i(C_\bullet)=\ker \partial_{i} / \mathop{\text{Im}} \partial_{i+1}$, while the $i$-th \textit{cohomology group} is defined to be ${H^i(C_\bullet)= \ker \delta_{i+1} / \mathop{\text{Im}} \delta_{i}}$. 
The \textit{distance} of a homology group $H_i(C_\bullet)$ is the minimum weight over all elements in $\ker \partial_i \backslash \mathop{\text{Im}} \partial_{i+1}$, and similar for cohomology groups. 
A pair of spaces $C_{i-1},C_{i},$ connected by boundary and coboundary maps $\partial_{i},\delta_{i}$, can be viewed as a \textit{hypergraph} with vertices and hyperedges given by basis vectors of $C_{i-1}$ and $C_i$, respectively.

A CSS code can be described by a length-three chain complex, as written above, where $C_0$ correspond to the $X$-checks, $C_1$ correspond to qubits, and $C_2$ correspond to $Z$-checks. 
The group of logical $Z$ ($X$) operators is isomorphic to the first homology (cohomology) group of this chain complex. 
The following representation of $C_1$ sends elements of $\ker \delta_2$ to $X$-type logical operators 
\begin{align}
    c \mapsto X(c) := \prod_{q} X_q^{c_q}.
\end{align}
Above, the product is over all qubits $q$, which label the basis states of $C_1$, and each coefficient $c_q\in \mathbb{F}_2$ takes a binary value, 0~or~1. 
This is similar for $\ker \partial_1$ and $Z$-type logical operators. 
The weight of an operator under this representation is given by the Hamming weight of the corresponding binary vector in $C_1$. 
The qLDPC condition is satisfied iff $\delta_1,\partial_2,$ are sparse. 
The $Z$ ($X$) distance of the CSS code equals the distance of the first homology (cohomology) group.

\section{Higher-form transversal gates}

\begin{figure}[t]
    \centering
    \includegraphics[page=8]{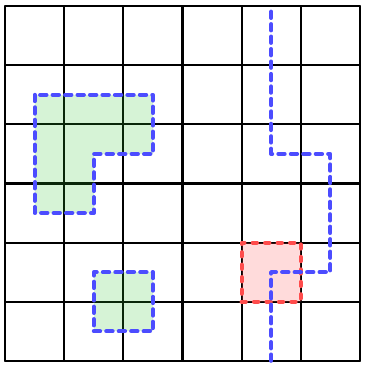}
    \caption{A 1-form transversal gate with chain complex given by the cellulation of a torus shown above. 
    Cocycles in $\ker \delta_2$ (depicted as blue dashed lines) determine the support of elements in the 1-form transversal gate. 
    Vertices (shaded green) are generators of the logically trivial subgroup of the 1-form transversal gate $\mathop{\text{Im}} \delta_1$.
    Plaquettes (shaded red) generate the parity constraints on edges (red dashed lines), corresponding to $\partial_2$, that specify $\ker \delta_2$. }
    \label{fig:0}
\end{figure}

In this section we define higher-form transversal gates for qLDPC codes. 
A conventional $h$-form symmetry is associated with a group generated by operators that act on codimension-$h$ submanifolds of a finite-dimensional topological space~\cite{kapustin2013higher,Gaiotto2015}. 
In this terminology, a 0-form symmetry is a standard global symmetry. 
For a chain complex derived from a cellulation of a manifold, the group of $h$-cocycles $\ker \delta_{h+1}$ provides an example of an $h$-form symmetry. 
Equivalently, by Poincar\'e duality, for a $D$-dimensional manifold the group of $(D-h)$ cycles $\ker \partial_{D-h}$ provides another example of an $h$-form symmetry. Fig.~\ref{fig:0} depicts elements of the 1-form symmetry associated to a square-lattice cellulation of the torus. 

\subsection{$h$-form transversal gates}

Inspired by the conventional definition, we now propose a definition of higher-form transversal gates for qLDPC codes. 
Our definition reduces to the conventional notion of higher-form symmetry for topological codes. 
Proposals for generalizations of higher-form symmetry have previously appeared in the context of subsystem symmetries~\cite{Rayhaun2021}. 

\begin{definition}[$h$-form transversal gate]
    Given a chain complex 
    \begin{align}
    \label{eq:HFCC}
     C_0 \mathop{\leftrightarrows}^{\partial_{1}}_{\delta_{1}} 
    \dots
    C_{h-1} \mathop{\leftrightarrows}^{\partial_h}_{\delta_h} 
    C_h \mathop{\leftrightarrows}^{\partial_{h+1}}_{\delta_{h+1}} 
    C_{h+1} 
    \dots
    \end{align}
    an $h$-form transversal gate is defined to be a representation of the group $\ker \delta_{h+1}$ of the form 
    ${U(c):= \prod_{s} U_s(c_s)}$, for $c\in \ker \delta_{h+1}$. 
    Here,  $s$ labels the locations of sites which may contain multiple qubits, and $U_s$ is an on-site representation of $\mathbb{F}_2$. 
\end{definition}

Our main focus is on representation of $\mathbb{Z}_2$ and so we abuse notation and simply write $U_s$ in place of $U_s(1)$, while  $U_s(0)=\mathds{1}$.
A simple example of a 1-form transversal gate, with on-site representation $U_s=X_s$, is given by the $X$ logical operators of a CSS code.

In this work, we consider a loose notion of \textit{transversal gate} and \textit{on-site representation}. 
Technically, we require the higher-form symmetry to be \textit{anomaly free}, in the sense that it can be gauged.
Unpacking this, the operators $U_s$ are required to commute but they may be different for each site, they may act on multiple qubits near the site $s$ and their support may overlap. 
The logical action may also be distinct from the physical on-site actions. 
We require the support of the $U_s$ operators to be sparse, i.e.~each operator acts on a constant number of qubits, and each qubit is acted upon by a constant number of $U_s$ operators. 
We define a \textit{strongly} transversal gate to be one where all on-site actions are the same unitary $U$, and furthermore the action on each logical qubit is $U$ in some logical basis. 
 
For example, the group of $X$ type logicals, or $Z$ type logicals, are strongly transversal and nonanomalous i.e.~they can be gauged on any stabilizer code~\cite{Williamson2016,kubica2018ungauging,Rakovszky2023,Cowtan2025Fast}. 
However, the group of all $X$ and $Z$ type logical cannot be simultaneously gauged when there are nontrivial logical operators. 
This is due to a mixed anomaly between the $X$ and $Z$ type higher-form symmetries which is a consequence of the nontrivial commutation relations between the logical operators. 
In this case the anomaly is also reflected by the fact that the on-site representation generated by both $X$ and $Z$ forms a projective representation of $\mathbb{Z}_2$, whereas we require $U_s$ to form a standard linear representation for the higher-form transversal gate to be anomaly free.

We require any family of qLDPC codes we consider to have a growing distance, i.e.~there are no constant weight logical operators.
Similarly, we require the higher-form transversal gates we consider to be qLDPC in the sense that they admit sparse $\delta_\bullet,\partial_\bullet,$ matrices. 
We also require the homology (cohomology) groups of the higher-form transversal gate to have growing distances. 
This means that there are no constant weight transversal gates that are homologically (cohomologically) trivial. 
Equivalently, these conditions ensure that the higher-form transversal gates are cleanable to the complement of any constant-sized region. Here, regions are defined using the locality inherited from the \textit{Tanner graph} of the qLDPC code. 
This implies that, for any element in $\ker \delta_{h+1}$ from an $h$-form transversal gate and any constant-sized region in a qLDPC code, there is a local element, in $\mathop{\text{Im}}\delta_{h}$ of the $h$-form transversal gate, which has the same action when restricted to the constant-sized region. 
Hence, if a local operator commutes with the locally generated subgroup $\mathop{\text{Im}}\delta_{h}$ of an $h$-form transversal gate, it must commute with the full $h$-form transversal gate group $\ker \delta_{h+1}$. 

\subsection{1-form Clifford gates}


An important class of higher-form transversal gates are 1-form Clifford gates. 
In this work, we show that such gates can be measured in constant time to prepare logical magic states via parallel logical Clifford measurement. 
An advantage of this approach over transversal 0-form gates from the third level of the Clifford Hierarchy is that the application of 1-form Clifford gates can be treated wholly within the stabilizer formalism~\cite{gottesman1997stabilizer}. 

For the 1-form Clifford gate examples presented in this section we assume that the distance of the first cohomology of each chain complex is growing with the code size. 
This ensures that elements of $\ker \delta_2$ are locally cleanable. which implies that for any element in $\ker \delta_2$, and any constant sized region $R$, there is a local element in $\mathop{\text{Im}}$ that takes the same form on the region $R$. 
In this case, the conditions that a transversal gate preserves the code space can be stated in terms of cohomologically trivial elements. 

\subsubsection{1-form $CZ$ gates}

We first consider 1-form $CZ$ gates on CSS codes. 
Measuring a strongly transversal gate on a CSS code initialized in the $X$-basis produces $CZ$ magic states. 
Any CSS code supporting a 0-form transversal $CCZ$ gate also supports a 1-form $CZ$ gate. 
We define the \textit{group commutator} $c(A,B):=ABA^\dagger B^\dagger$ for unitary matrices $A,B.$ 
The 1-form $CZ$ gate is isomorphic to the group of $X$-type logical operators, obtained by taking the group commutator of each logical operator with the transversal $CCZ$ gate. 

More generally, a pair of CSS codes can support a 1-form transversal $CZ$ gate without a 0-form transversal $CCZ$ gate. 
We consider three chain complexes, the first $C^{(0)}_\bullet$ defines the 1-form transversal gate with an on-site $CZ$ action, the second and third $C^{(1)}_\bullet,C^{(2)}_\bullet,$ specify a pair of CSS codes on which the gate acts
\begin{align}
    C_0^{(i)} \mathop{\leftrightarrows}^{\partial_{1}^{(i)}}_{\delta_{1}^{(i)}} 
    C_1^{(i)} \mathop{\leftrightarrows}^{\partial_2^{(i)}}_{\delta_2^{(i)}} 
    C_2^{(i)} 
    \label{eq:1FCSS}
\end{align}
for $i=0,1,2.$
We require an isomorphism between the vertices of the symmetry $C_1^{(0)}$ and a subset of the qubits in $C^{(1)}_\bullet,C^{(2)}_\bullet.$
For the transversal gate to preserve the codespace of the CSS codes it must satisfy
\begin{align}
    \mathop{\text{Im}} \delta_1^{(0)} \circ \mathop{\text{Im}} \delta_1^{(1)} \in \ker \partial_1^{(2)} ,
\end{align}
which can equivalently be formulated as
\begin{align}
    \mathop{\text{Im}} \delta_1^{(0)} \circ \mathop{\text{Im}} \delta_1^{(1)} \cdot \mathop{\text{Im}} \delta_1^{(2)} = 0 ,
\end{align}
where $\circ$ denotes the element-wise product, and $\cdot$ denotes the dot product, of $\mathbb{F}_2$-vectors. 
In the above equation, we consider the restriction of $\mathop{\text{Im}} \delta_1^{(1)},\mathop{\text{Im}} \delta_1^{(2)},$ to the qubits identified with $C_1^{(0)}$. 
The above expressions are symmetric, and hence any pair of CSS codes supports a 1-form $CZ$ gate. 
This follows by noting that $\mathop{\text{Im}} \delta_1^{(1)} \circ \mathop{\text{Im}} \delta_1^{(2)}$ specifies $\ker \delta_2^{(0)}$. 
We remark that this 1-form $CZ$ gate may be trivial. 

For a 0-form transversal $CCZ$ gate on three copies of a CSS code $C^{(1)}_\bullet$, the same complex describes the codes and the 1-form transversal gate. Hence, in the above notation we have $C^{(0)}_\bullet=C^{(1)}_\bullet=C^{(2)}_\bullet$ for the 1-form transversal $CZ$ gate obtained by taking group commutators between the transversal $CCZ$ gate and the $X$ logical operators on the third copy of the code. 

\subsubsection{1-form $XS$ gates}

We now consider 1-form $\sqrt{-i}XS~(\sqrt{i}XS^\dagger)$ gates on CSS codes. 
Below, we write $XS~(XS^\dagger)$ and leave the $\sqrt{\pm i}$ phases impicit. 
Measuring a strongly transversal $XS~(XS^\dagger)$ gate on a CSS code initialized in a Pauli basis will produce $T~(T^\dagger)$ magic states. 
Any CSS code supporting a 0-form transversal $T~(T^\dagger)$ gate also supports a 1-form $XS~(XS^\dagger)$ gate. 
This is obtained by conjugating the group of $X$-type logical operators by the transversal $T~(T^\dagger)$ gate. 

A CSS code can support a 1-form transversal $XS$ $(XS^\dagger)$ gate without a 0-form $T$ $(T^\dagger)$ gate. 
We consider two chain complexes, the first $C^{(0)}_\bullet$ defines the 1-form transversal gate, the second $C^{(1)}_\bullet,$ defines a CSS code on which the gate acts. 
For the 1-form gate to preserve the codespace requires
\begin{align}
    \mathop{\text{Im}} \delta_1^{(0)} \circ \mathop{\text{Im}} \delta_1^{(1)} \cdot \mathop{\text{Im}} \delta_1^{(1)} = 0 .
\end{align}
Additionally, for each region of intersecting support between an $X$-check and 1-form transversal gate, specified by ${\mathop{\text{Im}} \delta_1^{(0)} \circ \mathop{\text{Im}} \delta_1^{(1)}} $, we require that the parity of the number of on-site actions $XS$ equals the parity of half the total number of sites in that support. 
This is satisfied when all overlap regions are multiples of four and every on-site representation is $XS^\dagger$. 
This is also satisfied when there are an equal number of $XS$ and $XS^\dagger$ sites in each intersecting support region. 

For a 0-form transversal $T$ $(T^\dagger)$ gate, the same complex describes the CSS code and the 1-form $XS$ $(XS^\dagger)$ gate, i.e.~$C^{(0)}_\bullet=C^{(1)}_\bullet$.

\section{Higher-form gauging measurement}

\label{sec:HFGM}

\begin{figure}[t]
\begin{algorithm}[H]
\caption{Higher-form gauging measurement}
\label{alg:GaugeLogical}
\begin{algorithmic}
    \Require A qLDPC code $\mathcal{C}$ with an $h$-form transversal gate~$C_\bullet$. 
    An initial physical code state $\ket{\psi}$. 
    A choice of generators for the $h$-th cohomology group $H^{h}(C_\bullet)$ given by representatives $\{\ell_i\}_{i\in \mathcal{I}}$. 
    A function $F(M,x)$ that returns a solution $y$ to the $\mathbb{F}_2$-linear equation $x=M y$ via an efficient method such as Guassian elimination. 
    \vspace{.1cm}
    
    \Ensure The result $\sigma_i=\pm1$ of measuring each representative $\ell_i$ from a generating set for the $h$-th cohomology group $H^{h}(C_\bullet)$, indexed by $i\in \mathcal{I}$. 
    The post-measurement code state $\ket{\Psi}=G_\sigma\ket{\psi}:=\prod_{i\in \mathcal{I}}\frac{1}{2}(\mathds{1}+\sigma_i U(\ell_i))\ket{\psi}$. 
    \vspace{.1cm}

    \State $\{\sigma_i\}_{i\in \mathcal{I}} \gets \{1\}_{i \in \mathcal{I}}$
    \State $\ket{\Psi} \gets \ket{\psi}$
    \State $\{x_e\}_{e\in C_{h+1}} \gets \{0\}_{e\in C_{h+1}}$
    \ForEach{hyperedge $e$ in $C_{h+1}$} 
    \State $\ket{\Psi} \gets \ket{\Psi}\otimes \ket{0}_e$ 
    \Comment{{\footnotesize Initialize a qubit on $e$}}
    \EndFor
    \ForEach{vertex $v$ in $C_h$}
	\State Measure $A_v$ on $\ket{\Psi}$
        \Comment{{\footnotesize $A_v=U_v\prod_{e\in \delta_{h+1} v}X_e$}}
        \State $\varepsilon \gets $ Measurement result 
        \Comment{{\footnotesize Measurement result is $\pm 1$}}
        \State $\ket{\Psi} \gets \frac{1}{2}(\mathds{1}+\varepsilon A_v)\ket{\Psi}$
        \Comment{{\footnotesize Post-measurement state}}
        \ForEach{$i \in \mathcal{I}$}
            \If{$v \in \mathop{\text{supp}}U(\ell_i)$}
                \State $\sigma_i \leftarrow \varepsilon \sigma_i$
                \Comment{{\footnotesize Update logical measurement}}
            \EndIf
        \EndFor
    \EndFor
    \ForEach{hyperedge $e$ in $C_{h+1}$}
	\State Measure $Z_e$ on $\ket{\Psi}$
        \State $(-1)^{x_e}\gets $ Measurement result
        \Comment{{\footnotesize $x_e$ is~$0$ or $1$}}
        \State $\ket{\Psi} \gets \frac{1}{2}(\mathds{1}+(-1)^{x_e} Z_e)\ket{\Psi}$
        \Comment{{\footnotesize Post-measurement state}}
        \State Discard auxiliary qubit on $e$
    \EndFor
    \State $y \gets F(\delta_{h+1},x)$ 
    \State $\ket{\Psi} \gets  U(y) \ket{\Psi}$
    \Comment{{\footnotesize Apply byproduct operator}}
\end{algorithmic}
\end{algorithm}
\end{figure}


In this section we present our main result. 
The $h$-form gauging measurement is performed in three steps which are formalized in Algorithm~\ref{alg:GaugeLogical}. 
First, an ancilla qubit is initialized in the $\ket{0}$ state for each hyperedge $e$. 
These hyperedges correspond to basis vectors of $C_{h+1}$, which we denote $e\in C_{h+1}$.
Second, generalized Gauss law operators are measured for each vertex $v$ corresponding to basis vectors $v \in C_{h}$
\begin{align}
    A_v := U_v \prod_{e \in \delta_{h+1} v} X_e . 
\end{align}
Third, all ancilla hyperedge qubits are measured in the $Z$ basis and discarded.
Finally, a unitary byproduct operator, corresponding to a partial symmetry, is applied to the vertex qubits. 
This procedure can be viewed as an instance of the gauging measurement procedure from Ref.~\cite{williamson2024low} applied to the hypergraph associated with the boundary map $\partial_{h+1}$ and the potentially non-Pauli on-site representation specified by $U_v$. 
This procedure is highly parallelizable, including the qubit initialization, $A_v$ measurement, and qubit readout.

Performing the $h$-form gauging measurement, rather than 0-form gauging measurements on individual representatives for a generating set of cohomology classes $H^{h}(C_\bullet)$, comes with several advantages. 
First, the $h$-form gauging measurement can be performed fault-tolerantly in a constant number of time steps, while $0$-form measurement requires $d$ time steps~\cite{Cowtan2025Fast}. 
Second, the $h$-form gauging measurement performs the measurement of all cohomology classes $H^{h}(C_\bullet)$ in parallel. 
Third, the $h$-form gauging measurement requires only a linear qubit overhead, even for densely overlapping logical representatives for the cohomology classes $H^{h}(C_\bullet)$, while parallel $0$-form measurement schemes requires an $O(tW(\log t + \log W)$ qubit overhead to measure $t$ logical operators of maximum weight $W$~\cite{zhang2024time,cowtan2025parallel,Hsieh2025Simplified}. 

\begin{figure}[t]
    \centering
    \includegraphics[page=7]{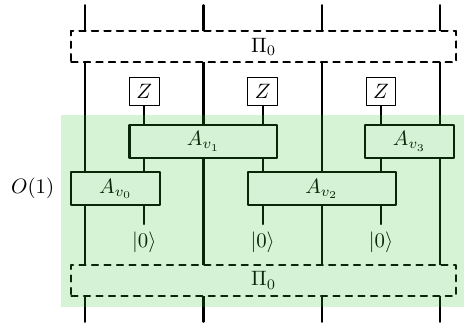}
    \caption{Schematic depiction of the higher-form gauging measurement. 
    Time runs up in this circuit diagram. 
    The box labelled $\Pi_0$ denotes an adaptive operation that prepares the code space in constant time. 
    The next time step shows $\ket{0}$ states being initialized on the hyperedge ancilla qubits. 
    Next, the boxes labelled $A_{v_i}$ denote the measurement of the Gauss's law terms that deform the code. 
    Fault-tolerance is ensured by the fact that products of $A_{v_i}$ terms in $\mathop{\text{Im}}\delta_{h}$ have eigenvalue $+1$ on the codespace. This results in detectors (inside the green shaded region) that can be used to correct measurement errors. 
    Next, the hyperedge qubits are measured out in the $Z$ basis and a byproduct operator is applied. 
    Finally, an adaptive operation is applied that measures the checks of the original code and applies a correction operator to restore the codespace. 
    For simplicity of presentation, the classical measurement results and feedforward operations are not explicitly depicted in this diagram.}
    \label{fig:A}
\end{figure}

\begin{figure}[t]
    \centering
    \includegraphics[page=6]{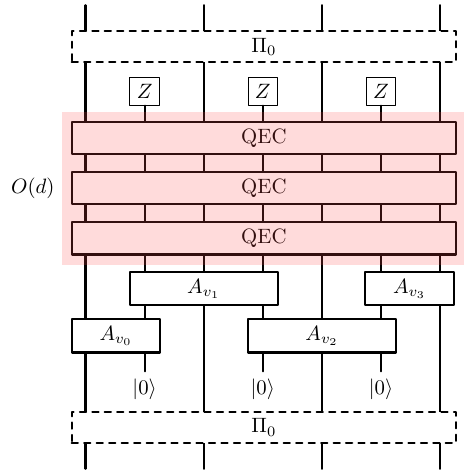}
    \caption{Schematic depiction of the standard gauging measurement~\cite{williamson2024low}. 
    The procedure and notation is similar to Fig.~\ref{fig:A}. 
    In contrast to the higher-form gauging measurement, no terms in $\mathop{\text{Im}}\delta_{h}$ are measured. This implies that $d$ rounds of syndrome extraction and error correction must be performed in the deformed code to ensure a distance against measurement errors (indicated by boxes labelled QEC)~\cite{Cowtan2025Fast}. The time savings of the higher-form gauging measurement are indicated by the red shaded region. 
    }
    \label{fig:B}
\end{figure}

We now present a basic result about the measurement that is achieved by higher-form gauging, followed by some remarks. 

\begin{theorem}[Higher-form gauging measurement]
    The higher-form gauging measurement procedure in Algorithm~\ref{alg:GaugeLogical} measures all $h$-form symmetry operators in $C_\bullet$, specified by $\ker{\delta_{h+1}}$, in parallel and in constant time
    \label{thm:1}
\end{theorem}
\begin{proof}
The $h$-form gauging measurement procedure for $C_\bullet$ applied to an initial state $\ket{\psi}$ results in the final (unnormalized) state 
\begin{align}
    &\bra{x} \prod_{v} (\mathds{1}+\varepsilon_v A_v) \ket{\underline{0}} \ket{\psi}
    \\
    &=
    \bra{x} \sum_{c \in C_h} \varepsilon(c) U(c) X(\delta_{h+1} c) \ket{\underline{0}} \ket{\psi}
    \\
    &= 
     \sum_{c \in C_h} \varepsilon(c)  \bra{x}X(\delta_{h+1} c) \ket{\underline{0}} U(c) \ket{\psi} .
\end{align}
Here, $\varepsilon_v=\pm$ denotes the random outcome of measuring $A_v$, $\varepsilon(c)=\prod_v \varepsilon_v^{c_v}$, $U(c)=\prod_v U_v^{c_v}$, $X(c)=\prod_v X_v^{c_v}$, for an $\mathbb{F}_2$-vector $c\in C_h$, and $\ket{\underline{0}}:=\ket{0}^{\otimes |C_{h+1}|}$.  
The final hyperedge measurement outcomes $x\in C_{h+1}$ must be of the form $x=\partial_{h+1} y$, for some $y\in C_h$, to have nonzero probability of occurring. 
Hence, we have
\begin{align}
     &\sum_{c \in C_h} \varepsilon(c)   \bra{\delta_{h+1} y}X(\delta_{h+1} c) \ket{\underline{0}} U(c) \ket{\psi} 
     \label{eq:HFGM1}
     \\
     &=
     \sum_{c \in C_h} \varepsilon(c)  \braket{\delta_{h+1} (c+y) | \underline{0}} U(c) \ket{\psi} 
     \\
     &=
     \sum_{c' \in C_h}  \varepsilon(c'+y) \braket{\delta_{h+1} c' | \underline{0}} U(c'+y) \ket{\psi} 
     \\
     &=
     \varepsilon(y) U(y) \sum_{c' \in \ker \delta_{h+1}} \varepsilon(c')  U(c') \ket{\psi}.
     \label{eq:HFGM2}
\end{align}
For a code with a growing distance and a sparse $h$-form transversal gate, local operators in $\mathop{\text{Im}} \delta_h$ stabilize the code space. 
In this case, we can further simplify the final state and write it as a projection onto a basis of representatives for the $h$-th cohomology group $H^h(C_\bullet)$
\begin{align}
    &U(y) \frac{1}{|H^h|} \sum_{[\ell]\in H^h} \varepsilon(\ell) U(\ell) \frac{1}{|C_{h-1}|}\sum_{b \in C_{h-1}} U(\delta_h b)
     \ket{\psi}
     \\
     &=
     U(y) \sum_{i \in \mathcal{I}} \frac{1}{2}(\mathds{1}+\varepsilon(\ell_i) U(\ell_i)) \ket{\psi}
     ,
\end{align}
where we have changed the overall normalization, dropped the irrelevant global phase factor $\varepsilon(y)$, introduced the notation $\ell_i,~i\in\mathcal{I},$ for a generating set of representatives for the $h$-th cohomology group, and used that $U(\delta_h b) \ket{\psi}=\ket{\psi}$. 
We have also used that $\varepsilon$ is constant on cohomology classes, as it is a homomorphism that satisfies $\varepsilon(c')=1$ for $c'\in\ker \delta_{h+1}$ when the codespace is stabilized by operators in $\mathop{\text{Im}} \delta_h$. 
\end{proof}

The higher-form gauging measurement in Algorithm~\ref{alg:GaugeLogical} can be generalized to allow hyperedge qubits initialized in a state $\ket{\ell_i}:=X(\ell_i)\ket{\underline{0}}$, for $[\ell_i]\in H^{h+1}(C_\bullet)$ with $\ell_i$ a minimum weight $(h+1)$st cohomology representative. For the trivial $(h+1)$st cohomology class, $[\underline{0}]$, this corresponds to the higher-form gauging introduced above. 
More generally, in Eq.~\eqref{eq:HFGM1}, we have $\ket{\ell_i}$ in place of $\ket{\underline{0}}$ and $\bra{\ell_i+\delta_{h+1}y}$ in place of $\bra{\delta_{h+1}y}$. 
However, this results in the same Eq.~\eqref{eq:HFGM2}. 

For a higher-form transversal gate with on-site symmetry actions that can be simultaneously mapped to Pauli $X$ operators under conjugation by a local unitary circuit, the $A_v$ terms can be mapped to single qubit $X$ operators via a further conjugation with the unitary circuit
\begin{align}
    \prod_{v} \prod_{e \ni v} C_v X_e.
\end{align}
Hence, after performing the combination of the above pair of local unitary circuits, each $A_v$ operator is mapped to a Pauli $X$ which fixes one qubit per vertex to be in the $\ket{+}$ state which can then be decoupled and discarded. 
This is sometimes referred to as disentangling the Gauss's law constraints~\cite{Gaugingpaper}. 

Next, we introduce some basic definitions and a result about the fault tolerance properties of the higher-form gauging measurement. 
To fully specify the deformed code that occurs during the higher-form gauging procedure, we consider the extension of the chain complex $C_\bullet$ in Eq.~\eqref{eq:HFCC} to $C_{h+2}$ 
\begin{align}
     C_0 \mathop{\leftrightarrows}^{\partial_{1}}_{\delta_{1}} 
    \dots
    C_{h-1} \mathop{\leftrightarrows}^{\partial_h}_{\delta_h} 
    C_h \mathop{\leftrightarrows}^{\partial_{h+1}}_{\delta_{h+1}} 
    C_{h+1} 
    \mathop{\leftrightarrows}^{\partial_{h+2}}_{\delta_{h+2}} C_{h+2}
    \dots .
\end{align}
In the context of the $h$-form gauging measurement, elements of $C_{h+2}$ correspond to generators of hyperedge cycles. 

\begin{definition}[Higher Cheeger constant]
    The $h$-Cheeger constant of a chain complex $C_\bullet$ is 
    \begin{align}
        \phi_h(C_\bullet):=\min_{c\in C_h } \max_{z \in \ker{\delta_{h+1}}} \frac{|\delta_{h+1} c|}{|c+z|}
        =
        \min_{\tilde{c}\in C_h } \frac{|\delta_{h+1} \tilde{c}|}{|\tilde{c}|}
    \end{align}
    where $\tilde{c}$ is the representative in the equivalence class $c+\ker \delta_{h+1}$ with minimal weight. 
\end{definition}

The above definition of the $h$-Cheeger constant reduces to the standard Cheeger constant for $h=0$. 
Next, we define gauged operators following the approach in Refs.~\cite{Bulmash2019Gauging,Prem2019Gauging}. 

\begin{definition}[Gauged operator]
    \label{def:GO}
    Consider a symmetric operator $S$, expanded in the basis of on-site symmetry charges
    \begin{align}
        S = \sum_{\chi \in \mathop{\text{Im}} \partial_{h+1}} S_\chi 
    \end{align}
    where 
    \begin{align}
        S_\chi := \sum_{c\in C_h} \chi(c) U(c) S U^{\dagger}(c)
        ,
    \end{align}
    which may evaluate to 0. 
    The deformed operator after gauging $\mathcal{G}[S]$ is 
    \begin{align}
        \mathcal{G}[S] := \sum_{\chi \in \mathop{\text{Im}} \partial_{h+1}} S_\chi Z(\varphi ) ,
    \end{align}
    where $\varphi$ is chosen to be a solution of $\chi = \partial \varphi$ with minimal weight. 
\end{definition}

The choice of solution $\varphi$ to the equation $\chi = \partial \varphi$ is arbitrary, however any choice is equivalent on the simultaneous $+1$ eigenspace of all $Z(c)$ operators for $c\in \ker \partial_{h+1}$. 
Furthermore, for any local operator $S$, the solution can be chosen to have support within a constant neighborhood of the operator's support. 
In this case, all choices are equivalent on the simultaneous $+1$ eigenspace of all $Z(\partial_{h+2} b)$ operators for $b\in C_{h+2}$. 

To gauge a nonsymmetric operator $S$, it is first projected onto its symmetric component
\begin{align}
    \tilde{S} = \sum_{c \in \ker \delta_{h+1}} U(c) S  U^\dagger(c) ,
\end{align}
and then the symmetrized operator $\tilde{S}$ is gauged following the above definition.

\begin{definition}[Gauged code]
    The deformed code in the higher-form gauging measurement is defined to be the simultaneous $+1$ eigenspace of vertex terms $A_v$, plaquette terms $B_p=\prod_{e \in \partial_{h+2} p} Z_e$, and deformed checks from the original code $\mathcal{G}[S_i]$.
\end{definition}

We now present our main result on the fault-tolerance of the higher-form gauging measurement. 
We consider an adversarial noise model where the code distance captures the minimum weight of a nontrivial logical fault on qubits while the measurement fault-distance captures the minimum weight of a nontrivial logical fault on measurement outcomes. 
\begin{theorem}[Fault-tolerance] 
\label{thm:main}
The higher-form gauging measurement is a fault-tolerant procedure, assuming a round of reliable stabilizer checks before and after. 
The code distance of the procedure is bounded by $ \phi_h(C_\bullet) d/2$, where $d$ is the original code distance.
The measurement fault-distance of the procedure is bounded by the distance of the $h$-th homology group $H_{h}(C_\bullet)$. 
\end{theorem}
We remark that the code distance bound can be improved to $ \phi_h(C_\bullet)d$ if the distance of $H^{h+1}(C_\bullet)$ is sufficiently large. 
The fault-tolerant higher-form gauging measurement procedure has constant time overhead and linear space overhead, i.e.~for codes with a constant rate, the constant rate property is preserved.
\begin{proof}
The higher-form gauging measurement proceeds in three steps as outlined above, preceded and followed by a round of reliable syndrome measurements for the original code which is assumed to be in the codespace (potentially after a correction operator has been applied). 
In this setting there are three types of errors. First, there are measurement errors on the $A_v$ operators. 
Second, there are $X$ errors on the hyperedge qubits which are equivalent to $\ket{0}$ state intialization errors and $Z_e$ measurement errors. 
Third, there are qubit errors on the vertex qubits. 

The full set of $A_v$ operators are being measured and hence we consider an error basis spanned by vertex error operators with definite on-site charge under $U_v$. 
This is analogous to the Pauli-error model for Pauli-stabilizer codes. 
The vertex errors can hence be commuted to just before the final reliable syndrome measurement, past the $A_v$ measurements and byproduct operator up to flipping the $A_v$ measurement outcome and a potential global phase factor, see Fig.~\ref{fig:A}. 
In the adversarial noise setting we consider, the byproduct operator is based on the $A_v$ outcomes after such measurement outcome flips are applied. 
This plays the role of error correction on the $A_v$ outcomes in this setting. 

We first consider the measurement errors during the higher-form gauging measurement. 
As we have assumed reliable measurements in the original code before and after the procedure, we restrict our attention to measurement errors of the $A_v$ checks. 
We remark that $\ket{0}$ initialization and $Z_e$ readout errors on the hyperedge qubits are equivalent to Pauli-$X$ errors, which we consider below.
Our focus is on a family of quantum codes with a growing distance and a sparse $h$-form symmetry.
Together, this implies that all coboundary symmetries stabilize the code space, i.e.~$U(\delta_h b)\ket{\psi}=\ket{\psi}$ for $c\in C_{h-1}$ and $\ket{\psi}$ in the code space, as these symmetries are generated by constant weight operators that preserve the code space. 
We have that 
\begin{align}
    A(c):=\prod_{v} A_v^{c_v} = U(c) ,
\end{align}
for any $c\in \ker \delta_{h+1}$. 
In particular, the product of $A_v$ checks on the support of any coboundary $\delta_h b$ stabilizes the codespace. 
Combined with the assumption of reliable stabilizer measurement preceding the higher-form gauging measurement, each such coboundary provides a meta-check on the $A_v$ measurement results as the associated product combined with the preceding reliable stabilizer measurements forms a detector. 
These detectors provide a parity check matrix on the measurement outcomes given by $\partial_h = \delta_h^T$. 
Logical measurement errors of the form $\partial_{h+1} c'$ are equivalent to conjugating the $A_v$ measurements by the operator $Z(c')$. 
Due to the initialization and readout in the $Z$ basis, such logical measurement errors are in fact spacetime stabilizers. 
They do not impact the logical measurement result of any higher-form transversal logical operator $U(c)$ for $c \in \ker \delta_{h+1}$. 
Hence, the distance to measurement errors is given by the distance of the $h$-th cohomology group $H_h(C_\bullet)$.  This is equivalent to the minimal weight of an element in $\ker \partial_h$ that is not in $\mathop{\text{Im}} \partial_{h+1}$. 

Next, we consider the distance against qubit errors during the higher-form gauging measurement.  
First, we consider $X$ errors on the hyperedges. 
Due to the local detectors from the $Z$ initialization and readout of the hyperedge qubits, the logical $X$ errors must correspond to cocycles in $\ker \delta_{h+2}$. 
The readout step also provides a measurement that determines an $(h+1)$-cohomology class for the hyperedge $X$ errors. 
For the cohomology class to be nontrivial, the weight of the error must be as large as the minimal weight of an element of $\ker \delta_{h+2} \backslash \mathop{\text{Im}} \delta_{h+1}$. 
If the $X$ error has lower weight, then it must be a coboundary $X(\delta_{h+1} c)$. 
Without loss of generality, we assume that $c$ is a solution with minimal weight. 
In this case it is equivalent to a vertex error
\begin{align}
    X(\delta_{h+1} c) \prod_{v} A_v^{c_v} = X(c) .
\end{align}
The relative weight satisfies $|X(\delta_{h+1} c)|/|X(c)| \geq \phi_h(C_\bullet)$. 
More generally, for an error in a nontrivial cohomology class $r\in [\ell_i]$, we have $r=\ell_i+\delta_{h+1} c$. 
Then, we have the equivalence of the hyperedge error
\begin{align}
    X(\ell_i+\delta_{h+1} c) \prod_v A_v^{c_v} = U(c) X(\ell_i). 
\end{align}
For the higher-form gauging procedure with hyperedges initialized in the state $\ket{\ell_i}$, this simply results in the error $U(c)$. 
The weight of the error multiplied by the reference cocycle satisfies 
\begin{align}
    2|X(r)|/|X(c)| &\geq (|X(r)|+|X(\ell_i)|)/|X(c)| 
    \\
    &\geq |X(r+\ell_i)|/|X(c)|
    \\
     &= |X(\delta_{h+1} c)|/|X(c)| 
     \\
     &\geq \phi_h(C_\bullet) .
\end{align}
Hence, cleaning an error from the hyperedges to the vertices results in an error that shrinks by a factor $\phi_h(C_\bullet)$ for a cohomologically trivial error, and $\phi_h(C_\bullet)/2$ for a cohomologically nontrivial error which must have a minimum weight set by $\ker \delta_{h+2} \backslash \mathop{\text{Im}} \delta_{h+1}$.

After cleaning the $X$ hyperedge errors, we are left with a logical error $f$ supported only on the vertices before the final code space measurement, see Fig.~\ref{fig:A}. 
That is, we have $\Pi_0 f G_\sigma \Pi_0$, where $\Pi_0$ is the projector onto the codespace (assuming a round of reliable syndrome extraction and feed-forward corrections in the original code) and $G_\sigma$ is the projection operator given by the higher-form gauging measurement in Algorithm~\ref{alg:GaugeLogical}. 
Since the codespace commutes with the higher-form transversal gate we have 
\begin{align}
    \Pi_0 f G_\sigma \Pi_0 = \Pi_0 f \Pi_0 G_\sigma ,
\end{align}
hence for $f$ to be a nontrivial logical error it must have distance at least $d$ (the distance of the original code) as it is projected onto the codespace. 

Hence the original logical error $F$ satisfies
\begin{align}
    |F|\geq \phi_h(C_\bullet)/2 |f| \geq \phi_h(C_\bullet) d /2,
\end{align}
where the first inequality comes from cleaning the error off the hyperedges, resulting in a vertex-only error $f$, and the second inequality comes from the distance bound for~$f$.
\end{proof}

In the standard gauging measurement procedure, additional ancilla systems were introduced to sparsify cycles of a graph, leading to a worst-case overhead of $O(W \log^3 W)$ for a logical of weight $W$~\cite{williamson2024low} (this can be reduced to $O(W \log W)$~\cite{Hsieh2025Simplified,ParsimoniousSurgery}). 
In the higher-form gauging measurement, the cycle checks are not directly measured, so it is unclear if a similar approach would be valuable. The procedure to perform such a sparsification for cycles of a hypergraph is also unclear. 
It is, in principle, possible to improve the measurement fault-distance by performing repeated rounds of deformed code check measurements. However, this requires being able to measure the deformed checks of the gauged code, including the hyperedge cycle checks.

For CSS codes with Z-basis preserving logical operators such as $CZ$ and $XS$, the $Z$ correction can be deferred as it commutes with the gauging operator.
In this case only reliable $Z$-type check measurements in the original code are required for a fault-tolerant higher-form gauging procedure. 
A consequence of this looser condition is that residual $Z$ errors may be present in the state produced by the gauging measurement which must be corrected later. 
This may prove advantageous as it essentially amounts to gaining a quantum resource saving at the cost of a more complicated decoding problem.

If the higher-form gauging measurement is applied to the group of $X$  ($Z$) logical operators of a CSS code the result is a parallel readout of all $X$ ($Z$) logical operators and the deformed code is equivalent to a product state, up to a local unitary. 
In this case the chain complex describing the transversal 1-form Pauli gate is identical to the chain complex describing the CSS code. 
This gauging map for CSS qLDPC codes was discussed in Refs.~\cite{Williamson2016,kubica2018ungauging,shirley2018FoliatedFracton,Rakovszky2023}.
If the on-site representation of this symmetry group is altered to include a product of $X$ ($Z$) operators across multiple codeblocks, the result is a readout of transversal product $X$ ($Z$) logical operators across the codeblocks. 
The $X$ ($Z$) logical group can also be truncated to a subregion to achieve a partial readout. 
This recovers many of the fast and fault-tolerant parallel measurement procedures in Ref.~\cite{Cowtan2025Fast}, including full block-reading and partial block-reading. 

\section{Examples}
\label{sec:Eg}


In this section we discuss examples of 1-form transversal Clifford gates. 
Any CSS code with a 0-form transversal non-Clifford gate, from the third level of the Clifford hierarchy, also supports a 1-form transversal Clifford gate. 
In particular, a 0-form gate from the third level, with on-site actions $U_v$ that preserve the $Z$ basis, defines a 1-form Clifford gate via conjugation of the $X$-type logical operators of the CSS code. 
This results in a transversal 1-form Clifford gate defined by the same chain complex as the CSS code itself, see Eq.~\eqref{eq:1FCSS} below, with on-site representation $U_vX_vU_v^\dagger$ (more generally, there may be multiple on-site representations if multiple qubits are grouped onto a single site, with one generator for the $X$ operator on each qubit).
In this case, the $B_p$ checks of the gauged code are isomorphic to meta-checks of the CSS code~\cite{campbell2019theory}, which are captured by the $\partial_3$ map in the longer chain complex below
\begin{align}
    C_0 \mathop{\leftrightarrows}^{\partial_{1}}_{\delta_{1}} 
    C_1 \mathop{\leftrightarrows}^{\partial_2}_{\delta_2} 
    C_2 \mathop{\leftrightarrows}^{\partial_3}_{\delta_3} C_3.
    \label{eq:1FCSS}
\end{align}
The measurement fault-distance of the 1-form gauging measurement is the distance of the first homology group $H_1(C_\bullet)$, see Theorem~\ref{thm:main}, which in this case corresponds to the $Z$-distance of the original CSS code. 
The distance of the second cohomology group $H^2(C_\bullet)$ which controls the improved distance bound $\phi_1 (C_\bullet) d$ for the 1-form gauging measurement, see below Theorem~\ref{thm:main}, in this case corresponds to the $Z$ meta-check distance of the original CSS code.

This is similar for 0-form gates from the third level that preserve the $X$ basis, and conjugate $Z$-type logical operators.

\subsection{3D Color Code}

The above setting includes the well-known example of the three-dimensional Color Code~\cite{Bombin2007exact}. 
This code is defined on a simplicial complex with 4-colorable vertices, $\{r,g,b,y\}$ see Fig.~\ref{fig:1}, or equivalently the dual of such a complex. 
A qubit is located on each tetrahedra of the primal complex, and the checks are given by 
\begin{align}
    X(v)&:= \prod_{t \ni v} X_t,
    \\
    Z(e)&:= \prod_{t \ni e} Z_t.
\end{align}
There is a meta-check associated to each vertex, generated by the product of adjacent edge stabilizers 
\begin{align}
    \prod_{e\ni v} Z(e)=\mathds{1}.
\end{align}
The chain complex in Eq.~\eqref{eq:1FCSS} for this example has $C_0$ as the space of $X$ type vertices, $C_1$ as the space of tetrahedra, $C_2$ as the space of edges, and $C_3$ as the space of meta-check vertices. 
The boundary (coboundary) matrices are specified by the incidence relations of the relevant objects, e.g.~vertices and tetrahedra for $\partial_1$ ($\delta_1$). 

The tetrahedra of the 4-colored simplicial complex can be bipartitioned into two sets, black $\Lambda_b$ and white $\Lambda_w$. 
The 0-form third-level gate on the Color Code is generated by 
\begin{align}
    T(\Lambda):=\prod_{t\in \Lambda_b} T \prod_{t\in \Lambda_w} T^\dagger.
\end{align}
This gate conjugates the $X$-type logical operators of the Color Code to produce a 1-form Clifford symmetry
\begin{align}
    U(c)=\prod_{t\in \Lambda_b} (\sqrt{i}XS^\dagger)^{c_t}_t \prod_{t\in \Lambda_w} (\sqrt{-i}XS)^{c_t}_t
\end{align}
for $c\in \ker{\delta_2}$. 
These cocycles include membranes given by collections of tetrahedra adjacent to edges of a certain color pair, such as $rg$~\cite{PhysRevB.91.245131}. 

The gauging measurement is performed by adding $\ket{0}$ qubits to the edges of the simplicial complex and measuring 
\begin{align}
    A_t=
    \begin{cases}
        (\sqrt{i}XS^\dagger)_t \prod_{e\in t} X_e, \qquad\ \text{for } t\in\Lambda_b, 
        \\
        (\sqrt{-i}XS)_t \prod_{e\in t} X_e, \qquad \text{for } t\in\Lambda_w,
    \end{cases}
\end{align}
before measuring $Z_e$ on all edge qubits, and applying a partial-symmetry byproduct operator, see Algorithm~\ref{alg:GaugeLogical}.
In this case, the $A_t$ and $\mathcal{G}[Z(e)]$ checks fully specify the deformed code. 
Applying a local unitary circuit
\begin{align}
    V:=T(\Lambda)^\dagger
    \prod_{t\in\Lambda_b} \prod_{e\in t} C_tX_e \,
    \prod_{t\in\Lambda_b} \prod_{e\in t} C_tX_e 
\end{align}
results in 
\begin{align}
    V A_t V^\dagger&= X_t ,\\
    V Z(e) Z^\dagger &= Z_e ,
\end{align}
which stabilizes a product state. 
Hence, in this case the higher-form gauging measurement achieves the same result as initializing the data qubits in $\ket{+}$, applying the 0-form transversal $T\, (T^\dagger)$ gate, and measuring the $Z$ checks of the Color Code using a standard syndrome extraction circuit with ancilla qubits on the edges. 

It is well-known that gauging the 1-form Clifford symmetry condenses a composite loop excitation consisting of a magnetic flux and a symmetry-protected topological excitation~\cite{PhysRevB.91.245131} which is a type of Cheshire charge string~\cite{else2017cheshire}. 
This results in a non-Pauli boundary to the trivial phase~\cite{PhysRevB.91.245131}. 
Hence, from a topological phase point of view, we expect the gauged code to be equivalent to a product state, up to a local unitary.

As discussed at the beginning of this section, the measurement fault-distance in this example is given by the $Z$-distance of the 3D Color Code which is stringlike. 
Since the deformed code corresponds to a product state that has no nontrivial logical operators, the code distance of the protocol is given by the distance of the 3D Color Code. 

The results in Ref.~\cite{Zhu2023} establish that this 1-form gauging measurement can produce magic states at a finite rate when applied to 3D Color Codes on appropriate cellulations of hyperbolic manifolds. 

A similar construction can be made with three copies of the toric code and a transversal CCZ gate which generates 1-form CZ gates~\cite{Vasmer2019three}. However, in this case we do not expect the gauged code to be equivalent to a product state.

\begin{figure}[t]
    \centering
    \includegraphics[page=1]{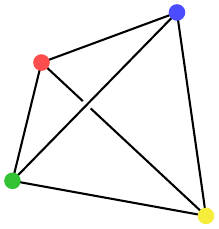}
    \caption{A tetrahedron with 4-colored vertices $\{r,g,b,y\}$. The edges inherit a 2-coloring from the adjacent edges. The faces inherit a color given by the vertex not included in the face. }
    \label{fig:1}
\end{figure}

\subsection{Twisted higher-group gauge theory}
\label{sec:EGB}

We now move on to our main example which supports a 1-form $CZ$ gate without requiring a 0-form $CCZ$ gate. 
The quantum code in this example is derived from the $\mathbb{Z}_2^{(0)}\times \mathbb{Z}_2^{(0)} \times \mathbb{Z}_2^{(1)}$ higher-form symmetry-protected topological order in Ref.~\cite{yoshida2015topological} by gauging the $\mathbb{Z}_2 \times \mathbb{Z}_2$ 0-form symmetry and applying a local unitary disentangler. 
The remaining 1-form symmetry is then gauged to implement a higher-form gauging measurement.
This is a generalization of the (2+1)D twisted quantum double model considered in Ref.~\cite{yoshida2015topological} to a (1-form)$\times$(1-form)$\times$(2-form) \textit{higher-group gauge theory} (HGGT). 
Following Ref.~\cite{davydova2025universal}, we apply this higher-group twisted gauge theory to measure logical Clifford operators via the gauging measurement. 

The initial code we consider is defined on any simplicial complex with 4-colored vertices $\{r,g,b,y\}$, see Fig.~\ref{fig:1}. 
A qubit is assigned to each red and green face only. 
The color of a face is determined by the vertex-color it does not contain, e.g.~a $y$ face has $r,g,b,$ vertices. 
Similarly, each edge $e_{ij}$ has two colors corresponding to the adjacent vertices $i,j\in \{r,g,b,y\}$, with $i\neq j$. 
The checks are 
\begin{align}
    X(v_{r})&:=\prod_{f_{r}\in \mathop{\text{lk}}(v_r)} X_{f_{r}},
    \\
    X(v_{g})&:=\prod_{f_{g}\in \mathop{\text{lk}}(v_g)} X_{f_{g}},
    \\
    Z_r(e_{ij})&:= \prod_{f_r \ni e_{ij}} Z_{f_r}, \qquad \text{for } i,j\neq r,
    \\
    Z_g(e_{ij})&:= \prod_{f_g \ni e_{ij}} Z_{f_g}, \qquad \text{for } i,j \neq g ,
    \label{eq:EG2CSS}
\end{align}
where lk denotes the \textit{link} of a vertex. 
The {link} of an  $i$-simplex $\triangle_i$, in a simplicial complex, can be thought of as the closure of a collection of $(D-i)$ simplices that wrap around the $i$-simplex. 
The link is defined as 
\begin{align}
    \mathop{\text{lk}} (\triangle_i) := \overline{\mathop{\text{st}}(\triangle_i)}\,\backslash\mathop{\text{st}}\big(\overline{\triangle_i}\big),
\end{align}
where st denotes the \text{star} of a simplex
\begin{align}
    \mathop{\text{st}}(\triangle_i):= \bigcup_{\triangle_j':\triangle^{(i)}\in \overline{\triangle_j'}} \triangle_j',
\end{align}
and an overline denotes the closure of a set of simplices, which is the smallest subcomplex containing them. 

In the initial code, there are two meta-checks assigned to each yellow and blue vertex, formed by the product of adjacent red or green edge checks
\begin{align}
    \prod_{e_{ij} \ni v_i} Z_r(e_{ij})
    &= \mathds{1}, \qquad \text{for } i=b,y,\ j\neq r,
    \\
    \prod_{e_{ij} \ni v_i} Z_g(e_{ij})
    &= \mathds{1}, \qquad \text{for } i=b,y,\ j\neq g.
\end{align}
This code corresponds to a pair of decoupled 3D surface codes on red and green sublattices which are formed by connecting pairs of red vertices through adjacent red faces, and similarly for the green sublattice. 
More specifically, the red sublattice consists of edges joining pairs of red vertices that share a red face in the intersection of their links, and similarly for the green sublattice. 

The 1-form $CZ$ gate is specified by a chain complex
\begin{align}
    C_0 \mathop{\leftrightarrows}^{\partial_{1}}_{\delta_{1}} 
    C_1 \mathop{\leftrightarrows}^{\partial_2}_{\delta_2} 
    C_2 \mathop{\leftrightarrows}^{\partial_3}_{\delta_3} C_3,
    \label{eq:2FCSS}
\end{align}
where the blue and yellow vertices of the simplicial complex form a basis of $C_0$, the $by$-edges form a basis of $C_1$, the $rg$-edges form a basis of $C_2$, and the red and green vertices form a basis of $C_3$. 
The map $\delta_1$ sends a blue or yellow vertex to its adjacent $by$-edges, $\delta_2$ sends a $by$-edge to the adjacent $rg$-edges in its link (i.e.~those it shares tetrahedra with), $\delta_3$ sends an $rg$-edge to its adjacent red and green vertex. 
The 1-form gates $U(c)=\prod_{e_{by}} U_{e_{by}}^{c_{e_{by}}}$ are generated by 1-cocycles $c\in \ker \delta_2$ which correspond to 2-manifolds on the $rg$ sublattice, which is Poincar\'e dual to the $by$ sublattice. 
Importantly, the chain complex that specifies the 1-form transversal gate is distinct from the chain complex that defines the CSS code in Eq.~\eqref{eq:EG2CSS}. 

The on-site representations $U_{e_{by}}$ of the 1-form $CZ$ are associated to the edges of the $by$ sublattice. 
To define the on-site representation we first make an arbitrary choice of an orientation and a distinguished adjacent red face $f_0^{e_{by}}$ for each $by$-edge $e_{by}$. 
We suppress the superscript in $f_0^{e_{by}}$ and simply write $f_0$ below where the meaning is clear from context. 
The choice of distinguished face and orientation induces an ordering of the faces that are adjacent to each $by$-edge.
The ordering is found by starting from $f_0$ and following a dual path through the adjacent faces in the order given by the right hand rule applied to the orientation of $e_{by}$. 
The on-site representation is 
\begin{align}
    U_{e_{by}}= \prod_{\substack{f_r \ni e_{by} \\ f_r\neq f_0}} \prod_{f_g \in \gamma_{f_0 f_r}} C_{f_r}Z_{f_g} ,
    \label{eq:EG2U}
\end{align}
where $f_0$ is the red reference face adjacent to $e_{by}$ and $\gamma_{f_0 f_r}$ is the edge path on the green lattice from the green vertex in $f_0$ to the green vertex in $f_r$ following the orientation induced by $e_{by}$. 
We note that edges in the green lattice correspond to green faces in the original simplicial complex. 
Different choices of reference face lead to locally different 1-form symmetries that have the same logical action.

\begin{figure*}[t]
    \centering
    \includegraphics[scale=.6,page=2]{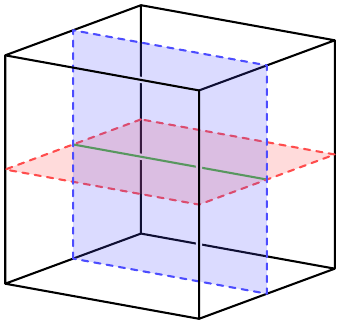}
    \quad
    \includegraphics[scale=.6,page=3]{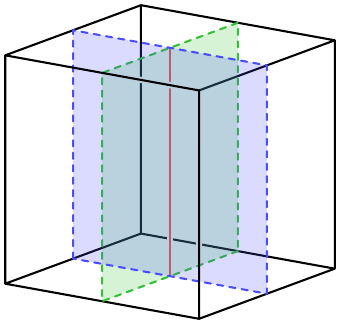}
    \quad
    \includegraphics[scale=.6,page=4]{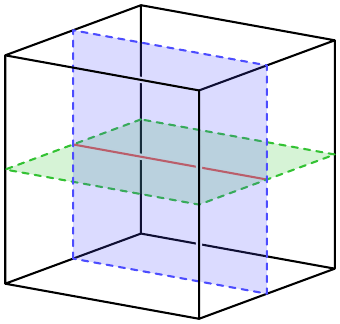}
    \quad
    \includegraphics[scale=.6,page=5]{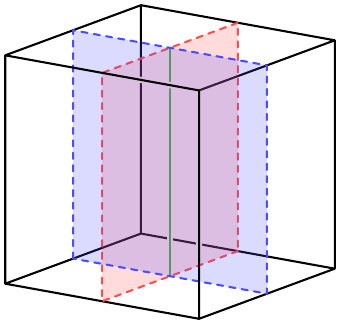}
    \caption{A schematic depiction of the logical action of the $\hat{x}$-membrane element of the 1-form transversal Clifford gate in example~\ref{sec:EGB} on the 3D torus. 
    The nontrivial element of a 1-form $CZ$ gate on an $\hat{x}$ membrane (blue shaded) conjugates a logical $X$ operator on one copy of surface code (red shaded membrane) to produce a logical $Z$ operator on the other copy (green line), and vice versa. }
    \label{fig:2}
\end{figure*}

\begin{lemma}
\label{lem:EG2}
    The action of the on-site symmetry on a membrane $m_{r}$ of red faces $f_r$ in the simplicial complex (similarly a membrane $m_g$ of green faces $f_g$) is
    \begin{align}
    U_{e_{by}} X(m_i) U_{e_{by}}^\dagger 
    & = Z_i(\gamma_{e_{by}f_0})
    \\
    &:=\prod_{f_i\in \gamma_{e_{by}f_0}} Z_{f_i} 
    \end{align}
    where $i=r,g,$ and $\gamma_{e_{by}f_0}$ is a dual path of faces adjacent to $e_{by}$. This path runs strictly between $f_r$ and $f_r'$ without including these faces. These are faces in the membrane $m_i$ that are adjacent to $e_{by}$, with $f_r<f_r'$ under the ordering induced by the reference face $f_0$ and the orientation of $e_{by}$. We define the path to be empty if $e_{by}\notin m$. 
\end{lemma}
\begin{proof}
    This follows from a direct application of CZ operators in Eq.~\eqref{eq:EG2U} to the $X$ operators on $f_r$ and $f_r'$. 
\end{proof}
In the above lemma, a red membrane $m_r$ corresponds to a 2-chain of only red faces, on the 4-colored simplicial complex, with either zero or two faces adjacent to each edge, and similar for green membranes. 

The above lemma can be applied to calculate the action of the on-site symmetry on more general 2-chains, on the 4-colored simplicial complex, by writing them as a product of membranes.
The operator $Z_r(\gamma_{e_{by}f_0})$ is a truncation of the check $Z_r(e_{by})$, and similar for $Z_g(\gamma_{e_{by}f_0})$ and $Z_g(e_{by})$. 
This is a small segment of red string operator that creates syndromes on the vertex stabilizers $X(v),X(v'),$ for $v\in f_r$, $v'\in f_r'$, with $v,v',\notin e_{by}$, and similar for green. 

\begin{theorem}
    The action of the 1-form $CZ$ gate for a cocycle $c\in \ker \delta_2$ on a membrane $m_i$, $i=r,g,$ in the simplicial complex is
    \begin{align}
        U(c) X(m_i) U(c)^\dagger &= Z_i(c \cap m_i) ,
        \\
        &:= \prod_{e_{by}\in c \cap m_i} Z_i(\gamma_{e_{by}f_0})
    \end{align}
    where $c\cap m_i$ denotes the collection of $e_{by}$-edges supported on the intersection of the 1-cocycle $c$ with the membrane $m_i$, and $\gamma_{e_{by}f_0}$ denotes the string operator segment associated to the action of $U_{e_{by}}$ on $X(m_i)$, see Lemma~\ref{lem:EG2}. 
\end{theorem}
\begin{proof}
    This follows by applying Lemma~\ref{lem:EG2} to $U_{e_{by}}$ on each $by$-edge in $c\cap m_i$.
\end{proof}
We now discuss particular applications of the above theorem. 
If $m_i$ is a homologically trivial membrane, such as $\mathop{\text{lk}} v_i$ for $X(v_i)$, $c \cap m_i$ is a trivial loop and $Z_i(c \cap m_i)$ is a stabilizer. 
In contrast, if $c \in \ker{\delta_2} \backslash \mathop{\text{Im}} \delta_1$ and $m_i$ is homologically nontrivial then $c \cap m_i$ may be a nontrivial loop and $Z_i(c \cap m_i)$ a nontrivial logical operator. 
For instance, on a 3D torus, for $c$ and $m_i$ corresponding to inequivalent embedded 2D tori,  $Z_i(c \cap m_i)$ corresponds to a nontrivial logical string operator, see Fig.~\ref{fig:2}. 

The logical action of the 1-form $CZ$ gate depends on the homology of the manifold on which the code is defined. 
For the 3D torus, the logical action of an $\hat{x}$-oriented $CZ$ membrane on a $\hat{y}$-oriented red $X$ membrane is a $\hat{z}$-oriented green $Z$ string. Similarly, the action on a $\hat{z}$-oriented red $X$ membrane is a $\hat{y}$-oriented green $Z$ string. The action on green $X$ membranes is analogous with the roles of red and green reversed. 
Hence, the logical action of the $\hat{x}$-oriented $CZ$ membrane is $C_{\hat{y}\textcolor{\red}{r}}Z_{\hat{z}\textcolor{\green}{g}} \, C_{\hat{y}\textcolor{\green}{g}}Z_{\hat{z}\textcolor{\red}{r}}$. 
Similarly, the logical actions of the  $\hat{y}$ and $\hat{z}$-oriented $CZ$ membranes are $C_{\hat{x}\textcolor{\red}{r}}Z_{\hat{z}\textcolor{\green}{g}} \, C_{\hat{x}\textcolor{\green}{g}}Z_{\hat{z}\textcolor{\red}{r}}$ and $C_{\hat{x}\textcolor{\red}{r}}Z_{\hat{y}\textcolor{\green}{g}} \, C_{\hat{x}\textcolor{\green}{g}}Z_{\hat{y}\textcolor{\red}{r}}$, respectively. 

We now discuss the 1-form gauging measurement for the $CZ$ gate in this example. 
To perform this 1-form gauging measurement we introduce $\ket{0}$ qubits onto the $rg$-edges, next we measure $A_{e_{by}}$ operators, and finally we measure $Z_{e_{rg}}$ operators followed by applying a partial symmetry byproduct operator. 
The checks of the deformed HGGT code are
\begin{align}
    A_{e_{by}} &= U_{e_{by}} \prod_{e_{rg}\in \mathop{\text{lk}}(e_{by})} X_{e_{rg}} ,
    \label{eq:HGGTC1}
    \\
    \mathcal{G}[X(v_{i})]&=V_{v_i}\prod_{f_{i}\in \mathop{\text{lk}}(v_i)} X_{f_{i}}, &\text{where } i=r,g,
    \\
    \mathcal{G}[Z_r(e_{ij})]&= Z_r(e_{ij}), &\text{where } i,j\neq r,
    \\
    \mathcal{G}[Z_g(e_{ij})]&= Z_g(e_{ij}), & \text{where } i,j \neq g ,
    \\
    Z(v_i) &= \prod_{e_{rg}\in \partial_3 v_i} Z_{e_{rg}}, & \text{where } i=r,g,
    \label{eq:HGGTC2}
\end{align}
where
\begin{align}
    V_{v_i} &= \prod_{\substack{v_j \in \mathop{\text{lk}}(v_i) \\ v_i\neq v_0}} \prod_{f_j \in \gamma_{v_0 v_j}} C_{e_{v_iv_j}}Z_{f_j} ,
    \label{eq:VertexCZs}
\end{align}
for $i,j\in\{r,g\}$ with $i\neq j$, the vertex $v_0\in \mathop{\text{lk}}(v_i)$ is an arbitrary reference vertex, $e_{v_iv_j}$ is the edge from $v_i$ to $v_j$, and $\gamma_{v_0 v_j}$ is a path on the $j=r,g,$ lattice from $v_j$ to $v_0$. 
There is a large amount of freedom involved in the specification of gauged $X$-type checks, as any choice that has the correct action on the $+1$-eigenspace of all $Z(\partial_3 v_i)$ checks leads to an equivalent code. 
This encompasses the choice of an arbitrary reference vertex $v_0$ for each red or green vertex $v_j$, $j=r,g.$

In this example, rather than apply the generic code deformation from Definition~\ref{def:GO}, we have used properties of measuring $X$-$CZ$-product operators on a CSS code. 
In particular, the $X$-type checks can be deformed by multiplying them with appropriate products of $CZ$ operators that are stabilizers in the initial code, as they act trivially on the hyperedge qubits initialized in $\ket{0}$ before the measurement of the $A_{e_{by}}$ operators.

The $A_{e_{by}}$ and $G[X(v_i)]$ operators only commute up to $Z$-type checks. 
Their group commutator is 
\begin{align}
    c(A_{e_{by}},\mathcal{G}[X(v_i)])
    &= Z_i(\gamma) 
\end{align}
for $i=r,g$, 
where $\gamma$ is a loop formed by the unions of paths from $v_j,v_j'\in \mathop{\text{lk}}(e_{by}) \cap \mathop{\text{lk}}(v_i),$ to $v_0$ see Eq.~\eqref{eq:VertexCZs}, and the path $\gamma_{e_{by}f_0}$ defined in Lemma~\ref{lem:EG2}. 
When the operators are not adjacent $\mathop{\text{lk}}(e_{by})\cap \mathop{\text{lk}}(v_i)= \emptyset $ and we have $\gamma=\emptyset,$ $Z(\gamma) =\mathds{1}$. 
Similarly, 
\begin{align}
        c(\mathcal{G}[X(v_r)],\mathcal{G}[X(v_g)])
    &= \begin{cases}
        \mathds{1} , &v_r\neq v_0',\, v_g\neq v_0,
        \\
        Z(v_r),  &v_r= v_0',\, v_g\neq v_0,
        \\
        Z(v_g),  & v_r\neq v_0',\, v_g= v_0,
        \\
        Z(v_r)Z(v_g),  & v_r= v_0',\, v_g= v_0,
    \end{cases} 
    \nonumber
\end{align}
where $v_0$ is the reference vertex for $v_r$, and $v_0'$ is the reference vertex for $v_g$. 

Gauging the 1-form $CZ$ gate measures Clifford logical operator representatives for all elements of the first cohomology group $H^{1}$, defined in Eq.~\eqref{eq:2FCSS}. 
When applied to a logical Pauli state, this can produce a magic state. 
For a logical $\ket{+}^{\otimes k}$ state, this produces a hypergraph magic state~\cite{zhu2023gates,zhu2025topological,Zhu2025b}. 
For the 3D torus topology shown in Fig.~\ref{fig:2}, this results in the following measurement-outcome-dependent logical projection 
\begin{align}
    (\mathds{1}\pm C_{\hat{y}\textcolor{\red}{r}}Z_{\hat{z}\textcolor{\green}{g}} \, C_{\hat{y}\textcolor{\green}{g}}Z_{\hat{z}\textcolor{\red}{r}})
    (\mathds{1}\pm & C_{\hat{x}\textcolor{\red}{r}}Z_{\hat{z}\textcolor{\green}{g}} \, C_{\hat{x}\textcolor{\green}{g}}Z_{\hat{z}\textcolor{\red}{r}})
    \nonumber \\
    &(\mathds{1}\pm C_{\hat{x}\textcolor{\red}{r}}Z_{\hat{y}\textcolor{\green}{g}} \, C_{\hat{x}\textcolor{\green}{g}}Z_{\hat{y}\textcolor{\red}{r}}) .
\end{align}
In particular for the $\ket{+_{\hat{x\textcolor{\red}{r}}}+_{\hat{x\textcolor{\green}{g}}}0_{\hat{y}\textcolor{\red}{r}}0_{\hat{y}\textcolor{\green}{g}}+_{\hat{z}\textcolor{\red}{r}}+_{\hat{z}\textcolor{\green}{g}}}$ state, we obtain a $CZ\otimes CZ$ state, which was discussed in Ref.~\cite{davydova2025universal}. 

On general topologies, the logical action of the 1-form transversal $CZ$ gate on a given surface applies a product of $CZ$ gates to pairs of red and green logical qubits whose logical $X$ membranes have a nontrivial triple intersection number with the $CZ$ surface. 
Similar to the 3D Color Code example, on suitable cellulations such as the hyperbolic manifolds considered in Ref.~\cite{Zhu2023}, this 1-form $CZ$ gauging measurement produces magic states at a finite rate. 

The measurement fault-distance of this 1-form $CZ$ gauging measurement is given by the distance of the first homology group $H_1(C_\bullet)$ from Eq.~\eqref{eq:2FCSS}. 
This is specified by the minimal weight topologically nontrivial $by$-edge path on the 4-colored simplicial complex. 
The code distance of this 1-form $CZ$ gauging measurement is specified by the twisted HGGT code specified by the checks in Eqs.~(\ref{eq:HGGTC1}--\ref{eq:HGGTC2}). 
This distance is determined by the shortest topologically nontrivial dual-path of red (green) faces, and the shortest topologically nontrivial $rg$-edge path. 
Here, the logical operator on the red (green) dual path is simply a product of Pauli-$Z$ operators, while the logical operator on the $rg$ path is a product of Pauli-$X$ operators and $CZ$ gates. 


\section{Discussion} 

In this work we have applied fast qLDPC surgery techniques to fault-tolerantly prepare logical magic states. 
There is a natural fit between the transversal structure of fast qLDPC surgery and parallel logical magic state preparation. 
We introduced the notion of higher-form transversal gates for quantum codes to formalize a set of suitable conditions for fast and fault-tolerant magic state preparation via logical measurement. 

Our main result is a higher form-gauging measurement procedure which has constant time overhead and linear qubit overhead. 
When applied to a constant-rate code this results in a fault-tolerant spacetime overhead proportional to the cost of the unencoded state preparation procedure, i.e.~constant multiplicative overhead. 
The fault distance of the higher-form gauging measurement procedure is determined by the homology of the chain complex that defines the higher-form transversal gate. 
In contrast to the standard gauging measurement~\cite{williamson2024low}, the higher-form gauging measurement is protected from measurement faults by detectors that are formed by local symmetries acting on the code space. 
When combined with single-shot fault-tolerant logical Pauli state preparation, gauging an appropriate transversal 1-form Clifford gate provides a single-shot logical magic state preparation procedure. 

We discussed several examples, including the well-known 3D Color Code which supports a 0-form transversal $T$ gate. 
We explained how CSS codes with 0-form third level gates lead to 1-form Clifford gates in general. 
Furthermore, we introduced an example based on twisted HGGT that was not derived from a 0-form third level gate. 
Finally, we formulated conditions on CSS codes that are necessary and sufficient to support a transversal 1-form $CZ$ or $XS$ $(XS^\dagger)$ gate. 

This points to an interesting research direction in code design for the purposes of fast high-rate logical magic state preparation. 
Rather than focusing on transversal non-Clifford gates, our work points to a path based on the less restrictive condition of transversal 1-form Clifford gates. 
This raises the challenge of constructing fixed-sized instances and asymptotic families of codes satisfying this condition. 
It also raises questions about more general higher form gates:  
Does a 0-form transversal gate from the $r$-th level of the Clifford hierarchy define a transversal $h$-form gate from the $(r-h)$-th level of the Clifford hierarchy, for all $0<h<r$? 
Are higher-form gates from higher levels of the Clifford hierarchy useful~\cite{Ellison2025a,Huang2025a,Warman2025}? 
Do 2-form, and higher form, transversal gates have advantages over 1-form gates for the purposes of fault-tolerant quantum computation?

We have left the decoding problem for the higher-form gauging measurement procedure to future work, as we expect its solution to depend on the family of codes to which the procedure is applied.  
We highlight that the decoding problem has a single-shot structure that makes it distinct in nature from the decoding problem for the standard gauging measurement procedure applied to 0-form Clifford gates~\cite{davydova2025universal,GaugingLDPC}. 
During the 0-form Clifford gauging measurement it is important to remain in a nonabelian code space for a number of timesteps that is proportional to the distance, see Fig.~\ref{fig:B}. 
The use of a nonabelian code space brings with it difficulties including the requirement of a just-in-time decoder~\cite{Bombin20182D,Brown2020FTNC,Scruby2022,Scruby2025,davydova2025universal,GaugingLDPC}.
In contrast, the higher-form gauging procedure spends only a single time step in a deformed nonabelian code. 
This avoids the need to measure deformed checks directly, and instead, relies on deforming back to the original code where the undeformed checks can be measured. 
For Pauli codes, this has the advantage of projecting qubit faults onto Pauli operators. 
We expect this to simplify the decoding problem.



\vspace{1cm}
\noindent \textbf{Note Added.} While this work was in progress, Ref.~\cite{Zhu2026NAQLDPC} appeared online. The construction in Ref.~\cite{Zhu2026NAQLDPC} shares similarities with our example in Section~\ref{sec:EGB}. However, the protocol in Ref.~\cite{Zhu2026NAQLDPC} uses $O(d)$ rounds while our protocol requires only a constant number of rounds. 

\section*{Acknowledgements} 
The author acknowledges inspiring discussions about 0-form Clifford gauging measurements with Yuanjie (Collin) Ren, Katie Chang, Anasuya Lyons, Harry Putterman, Nathanan Tantivasadakarn, Victor Albert, and Ben Brown, as part of a related ongoing collaboration~\cite{GaugingLDPC}. 
The author also acknowledges useful discussions with Alex Cowtan. 
Part of this work was done during the ``Noise-robust Phases of Quantum Matter'' program at the Kavli Institute for Theoretical Physics (KITP) and the ``Diving Deeper into Defects'' program at the Isaac Newton Institute for Mathematical Sciences, Cambridge (INI). 
DJW is supported by the Australian Research Council Discovery Early Career Research Award (DE220100625). 
This research was supported in part by grant NSF PHY-2309135 to the KITP and in part by EPSRC grant no EP/R014604/1 to the INI.

\bibliography{references.bib}

\begin{thebibliography}{98}%
\makeatletter
\providecommand \@ifxundefined [1]{%
 \@ifx{#1\undefined}
}%
\providecommand \@ifnum [1]{%
 \ifnum #1\expandafter \@firstoftwo
 \else \expandafter \@secondoftwo
 \fi
}%
\providecommand \@ifx [1]{%
 \ifx #1\expandafter \@firstoftwo
 \else \expandafter \@secondoftwo
 \fi
}%
\providecommand \natexlab [1]{#1}%
\providecommand \enquote  [1]{``#1''}%
\providecommand \bibnamefont  [1]{#1}%
\providecommand \bibfnamefont [1]{#1}%
\providecommand \citenamefont [1]{#1}%
\providecommand \href@noop [0]{\@secondoftwo}%
\providecommand \href [0]{\begingroup \@sanitize@url \@href}%
\providecommand \@href[1]{\@@startlink{#1}\@@href}%
\providecommand \@@href[1]{\endgroup#1\@@endlink}%
\providecommand \@sanitize@url [0]{\catcode `\\12\catcode `\$12\catcode `\&12\catcode `\#12\catcode `\^12\catcode `\_12\catcode `\%12\relax}%
\providecommand \@@startlink[1]{}%
\providecommand \@@endlink[0]{}%
\providecommand \url  [0]{\begingroup\@sanitize@url \@url }%
\providecommand \@url [1]{\endgroup\@href {#1}{\urlprefix }}%
\providecommand \urlprefix  [0]{URL }%
\providecommand \Eprint [0]{\href }%
\providecommand \doibase [0]{https://doi.org/}%
\providecommand \selectlanguage [0]{\@gobble}%
\providecommand \bibinfo  [0]{\@secondoftwo}%
\providecommand \bibfield  [0]{\@secondoftwo}%
\providecommand \translation [1]{[#1]}%
\providecommand \BibitemOpen [0]{}%
\providecommand \bibitemStop [0]{}%
\providecommand \bibitemNoStop [0]{.\EOS\space}%
\providecommand \EOS [0]{\spacefactor3000\relax}%
\providecommand \BibitemShut  [1]{\csname bibitem#1\endcsname}%
\let\auto@bib@innerbib\@empty
\bibitem [{\citenamefont {Yoder}\ \emph {et~al.}(2025)\citenamefont {Yoder}, \citenamefont {Schoute}, \citenamefont {Rall}, \citenamefont {Pritchett}, \citenamefont {Gambetta}, \citenamefont {Cross}, \citenamefont {Carroll},\ and\ \citenamefont {Beverland}}]{yoder2025tour}%
  \BibitemOpen
  \bibfield  {author} {\bibinfo {author} {\bibfnamefont {T.~J.}\ \bibnamefont {Yoder}}, \bibinfo {author} {\bibfnamefont {E.}~\bibnamefont {Schoute}}, \bibinfo {author} {\bibfnamefont {P.}~\bibnamefont {Rall}}, \bibinfo {author} {\bibfnamefont {E.}~\bibnamefont {Pritchett}}, \bibinfo {author} {\bibfnamefont {J.~M.}\ \bibnamefont {Gambetta}}, \bibinfo {author} {\bibfnamefont {A.~W.}\ \bibnamefont {Cross}}, \bibinfo {author} {\bibfnamefont {M.}~\bibnamefont {Carroll}},\ and\ \bibinfo {author} {\bibfnamefont {M.~E.}\ \bibnamefont {Beverland}},\ }\href {https://arxiv.org/abs/2506.03094} {\bibinfo {title} {Tour de gross: A modular quantum computer based on bivariate bicycle codes}} (\bibinfo {year} {2025}),\ \Eprint {https://arxiv.org/abs/2506.03094} {arXiv:2506.03094 [quant-ph]} \BibitemShut {NoStop}%
\bibitem [{\citenamefont {Panteleev}\ and\ \citenamefont {Kalachev}(2022)}]{panteleev2022asymptotically}%
  \BibitemOpen
  \bibfield  {author} {\bibinfo {author} {\bibfnamefont {P.}~\bibnamefont {Panteleev}}\ and\ \bibinfo {author} {\bibfnamefont {G.}~\bibnamefont {Kalachev}},\ }\bibfield  {title} {\bibinfo {title} {{Asymptotically good Quantum and locally testable classical {LDPC} codes}},\ }in\ \href {https://doi.org/10.1145/3519935.3520017} {\emph {\bibinfo {booktitle} {Proceedings of the 54th Annual ACM SIGACT Symposium on Theory of Computing}}},\ \bibinfo {series and number} {STOC 2022}\ (\bibinfo  {publisher} {Association for Computing Machinery},\ \bibinfo {address} {New York, NY, USA},\ \bibinfo {year} {2022})\ p.\ \bibinfo {pages} {375–388}\BibitemShut {NoStop}%
\bibitem [{\citenamefont {Breuckmann}\ and\ \citenamefont {Eberhardt}(2021{\natexlab{a}})}]{breuckmann2021balanced}%
  \BibitemOpen
  \bibfield  {author} {\bibinfo {author} {\bibfnamefont {N.~P.}\ \bibnamefont {Breuckmann}}\ and\ \bibinfo {author} {\bibfnamefont {J.~N.}\ \bibnamefont {Eberhardt}},\ }\bibfield  {title} {\bibinfo {title} {{Balanced Product Quantum Codes}},\ }\href {https://doi.org/10.1109/TIT.2021.3097347} {\bibfield  {journal} {\bibinfo  {journal} {IEEE Transactions on Information Theory}\ }\textbf {\bibinfo {volume} {67}},\ \bibinfo {pages} {6653} (\bibinfo {year} {2021}{\natexlab{a}})}\BibitemShut {NoStop}%
\bibitem [{\citenamefont {Leverrier}\ and\ \citenamefont {Zemor}(2022)}]{leverrier2022quantum}%
  \BibitemOpen
  \bibfield  {author} {\bibinfo {author} {\bibfnamefont {A.}~\bibnamefont {Leverrier}}\ and\ \bibinfo {author} {\bibfnamefont {G.}~\bibnamefont {Zemor}},\ }\bibfield  {title} {\bibinfo {title} {{Quantum Tanner codes}},\ }in\ \href {https://doi.org/10.1109/FOCS54457.2022.00117} {\emph {\bibinfo {booktitle} {2022 IEEE 63rd Annual Symposium on Foundations of Computer Science (FOCS)}}}\ (\bibinfo  {publisher} {IEEE Computer Society},\ \bibinfo {address} {Los Alamitos, CA, USA},\ \bibinfo {year} {2022})\ pp.\ \bibinfo {pages} {872--883}\BibitemShut {NoStop}%
\bibitem [{\citenamefont {Dinur}\ \emph {et~al.}(2022)\citenamefont {Dinur}, \citenamefont {Evra}, \citenamefont {Livne}, \citenamefont {Lubotzky},\ and\ \citenamefont {Mozes}}]{dinur2022locally}%
  \BibitemOpen
  \bibfield  {author} {\bibinfo {author} {\bibfnamefont {I.}~\bibnamefont {Dinur}}, \bibinfo {author} {\bibfnamefont {S.}~\bibnamefont {Evra}}, \bibinfo {author} {\bibfnamefont {R.}~\bibnamefont {Livne}}, \bibinfo {author} {\bibfnamefont {A.}~\bibnamefont {Lubotzky}},\ and\ \bibinfo {author} {\bibfnamefont {S.}~\bibnamefont {Mozes}},\ }\bibfield  {title} {\bibinfo {title} {{Locally testable codes with constant rate, distance, and locality}},\ }in\ \href {https://doi.org/10.1145/3519935.3520024} {\emph {\bibinfo {booktitle} {Proceedings of the 54th Annual ACM SIGACT Symposium on Theory of Computing}}},\ \bibinfo {series and number} {STOC 2022}\ (\bibinfo  {publisher} {Association for Computing Machinery},\ \bibinfo {address} {New York, NY, USA},\ \bibinfo {year} {2022})\ p.\ \bibinfo {pages} {357–374}\BibitemShut {NoStop}%
\bibitem [{\citenamefont {Gottesman}(2013)}]{gottesman2013fault}%
  \BibitemOpen
  \bibfield  {author} {\bibinfo {author} {\bibfnamefont {D.}~\bibnamefont {Gottesman}},\ }\bibfield  {title} {\bibinfo {title} {{Fault-tolerant quantum computation with constant overhead}},\ }\href@noop {} {\bibfield  {journal} {\bibinfo  {journal} {arXiv preprint arXiv:1310.2984}\ } (\bibinfo {year} {2013})}\BibitemShut {NoStop}%
\bibitem [{\citenamefont {Nguyen}\ and\ \citenamefont {Pattison}(2024)}]{nguyen2024quantum}%
  \BibitemOpen
  \bibfield  {author} {\bibinfo {author} {\bibfnamefont {Q.~T.}\ \bibnamefont {Nguyen}}\ and\ \bibinfo {author} {\bibfnamefont {C.~A.}\ \bibnamefont {Pattison}},\ }\bibfield  {title} {\bibinfo {title} {{Quantum fault tolerance with constant-space and logarithmic-time overheads}},\ }\href@noop {} {\bibfield  {journal} {\bibinfo  {journal} {arXiv preprint arXiv:2411.03632}\ } (\bibinfo {year} {2024})}\BibitemShut {NoStop}%
\bibitem [{\citenamefont {Tamiya}\ \emph {et~al.}(2024)\citenamefont {Tamiya}, \citenamefont {Koashi},\ and\ \citenamefont {Yamasaki}}]{tamiya2024polylog}%
  \BibitemOpen
  \bibfield  {author} {\bibinfo {author} {\bibfnamefont {S.}~\bibnamefont {Tamiya}}, \bibinfo {author} {\bibfnamefont {M.}~\bibnamefont {Koashi}},\ and\ \bibinfo {author} {\bibfnamefont {H.}~\bibnamefont {Yamasaki}},\ }\bibfield  {title} {\bibinfo {title} {{Polylog-time-and constant-space-overhead fault-tolerant quantum computation with quantum low-density parity-check codes}},\ }\href@noop {} {\bibfield  {journal} {\bibinfo  {journal} {arXiv preprint arXiv:2411.03683}\ } (\bibinfo {year} {2024})}\BibitemShut {NoStop}%
\bibitem [{\citenamefont {Nielsen}\ and\ \citenamefont {Chuang}(2012)}]{NielsenChuang}%
  \BibitemOpen
  \bibfield  {author} {\bibinfo {author} {\bibfnamefont {M.~A.}\ \bibnamefont {Nielsen}}\ and\ \bibinfo {author} {\bibfnamefont {I.~L.}\ \bibnamefont {Chuang}},\ }\href {https://doi.org/10.1017/cbo9780511976667} {\emph {\bibinfo {title} {{Quantum Computation and Quantum Information: 10th Anniversary Edition}}}}\ (\bibinfo  {publisher} {Cambridge University Press},\ \bibinfo {year} {2012})\BibitemShut {NoStop}%
\bibitem [{\citenamefont {Gottesman}(1997)}]{gottesman1997stabilizer}%
  \BibitemOpen
  \bibfield  {author} {\bibinfo {author} {\bibfnamefont {D.}~\bibnamefont {Gottesman}},\ }\emph {\bibinfo {title} {{Stabilizer codes and quantum error correction}}},\ \href {https://thesis.library.caltech.edu/2900/} {Ph.D. thesis},\ \bibinfo  {school} {California Institute of Technology} (\bibinfo {year} {1997})\BibitemShut {NoStop}%
\bibitem [{\citenamefont {Bravyi}\ \emph {et~al.}(2016)\citenamefont {Bravyi}, \citenamefont {Smith},\ and\ \citenamefont {Smolin}}]{bravyi2016trading}%
  \BibitemOpen
  \bibfield  {author} {\bibinfo {author} {\bibfnamefont {S.}~\bibnamefont {Bravyi}}, \bibinfo {author} {\bibfnamefont {G.}~\bibnamefont {Smith}},\ and\ \bibinfo {author} {\bibfnamefont {J.~A.}\ \bibnamefont {Smolin}},\ }\bibfield  {title} {\bibinfo {title} {{Trading Classical and Quantum Computational Resources}},\ }\bibfield  {journal} {\bibinfo  {journal} {Physical Review X}\ }\textbf {\bibinfo {volume} {6}},\ \href {https://doi.org/10.1103/physrevx.6.021043} {10.1103/physrevx.6.021043} (\bibinfo {year} {2016})\BibitemShut {NoStop}%
\bibitem [{\citenamefont {Cohen}\ \emph {et~al.}(2022)\citenamefont {Cohen}, \citenamefont {Kim}, \citenamefont {Bartlett},\ and\ \citenamefont {Brown}}]{cohen2022low}%
  \BibitemOpen
  \bibfield  {author} {\bibinfo {author} {\bibfnamefont {L.~Z.}\ \bibnamefont {Cohen}}, \bibinfo {author} {\bibfnamefont {I.~H.}\ \bibnamefont {Kim}}, \bibinfo {author} {\bibfnamefont {S.~D.}\ \bibnamefont {Bartlett}},\ and\ \bibinfo {author} {\bibfnamefont {B.~J.}\ \bibnamefont {Brown}},\ }\bibfield  {title} {\bibinfo {title} {{Low-overhead fault-tolerant quantum computing using long-range connectivity}},\ }\href {https://doi.org/10.1126/sciadv.abn1717} {\bibfield  {journal} {\bibinfo  {journal} {Science Advances}\ }\textbf {\bibinfo {volume} {8}},\ \bibinfo {pages} {eabn1717} (\bibinfo {year} {2022})}\BibitemShut {NoStop}%
\bibitem [{\citenamefont {Cross}\ \emph {et~al.}(2024)\citenamefont {Cross}, \citenamefont {He}, \citenamefont {Rall},\ and\ \citenamefont {Yoder}}]{cross2024improved}%
  \BibitemOpen
  \bibfield  {author} {\bibinfo {author} {\bibfnamefont {A.}~\bibnamefont {Cross}}, \bibinfo {author} {\bibfnamefont {Z.}~\bibnamefont {He}}, \bibinfo {author} {\bibfnamefont {P.}~\bibnamefont {Rall}},\ and\ \bibinfo {author} {\bibfnamefont {T.}~\bibnamefont {Yoder}},\ }\bibfield  {title} {\bibinfo {title} {{Improved {QLDPC} Surgery: Logical Measurements and Bridging Codes}},\ }\href@noop {} {\bibfield  {journal} {\bibinfo  {journal} {arXiv preprint arXiv:2407.18393}\ } (\bibinfo {year} {2024})}\BibitemShut {NoStop}%
\bibitem [{\citenamefont {Xu}\ \emph {et~al.}(2024)\citenamefont {Xu}, \citenamefont {Zhou}, \citenamefont {Zheng}, \citenamefont {Bluvstein}, \citenamefont {Ataides}, \citenamefont {Lukin},\ and\ \citenamefont {Jiang}}]{xu2024fast}%
  \BibitemOpen
  \bibfield  {author} {\bibinfo {author} {\bibfnamefont {Q.}~\bibnamefont {Xu}}, \bibinfo {author} {\bibfnamefont {H.}~\bibnamefont {Zhou}}, \bibinfo {author} {\bibfnamefont {G.}~\bibnamefont {Zheng}}, \bibinfo {author} {\bibfnamefont {D.}~\bibnamefont {Bluvstein}}, \bibinfo {author} {\bibfnamefont {J.}~\bibnamefont {Ataides}}, \bibinfo {author} {\bibfnamefont {M.~D.}\ \bibnamefont {Lukin}},\ and\ \bibinfo {author} {\bibfnamefont {L.}~\bibnamefont {Jiang}},\ }\bibfield  {title} {\bibinfo {title} {{Fast and Parallelizable Logical Computation with Homological Product Codes}},\ }\href@noop {} {\bibfield  {journal} {\bibinfo  {journal} {arXiv preprint arXiv:2407.18490}\ } (\bibinfo {year} {2024})}\BibitemShut {NoStop}%
\bibitem [{\citenamefont {Williamson}\ and\ \citenamefont {Yoder}(2024)}]{williamson2024low}%
  \BibitemOpen
  \bibfield  {author} {\bibinfo {author} {\bibfnamefont {D.~J.}\ \bibnamefont {Williamson}}\ and\ \bibinfo {author} {\bibfnamefont {T.~J.}\ \bibnamefont {Yoder}},\ }\bibfield  {title} {\bibinfo {title} {{Low-overhead fault-tolerant quantum computation by gauging logical operators}},\ }\href@noop {} {\bibfield  {journal} {\bibinfo  {journal} {arXiv preprint arXiv:2410.02213}\ } (\bibinfo {year} {2024})}\BibitemShut {NoStop}%
\bibitem [{\citenamefont {Ide}\ \emph {et~al.}(2024)\citenamefont {Ide}, \citenamefont {Gowda}, \citenamefont {Nadkarni},\ and\ \citenamefont {Dauphinais}}]{ide2024fault}%
  \BibitemOpen
  \bibfield  {author} {\bibinfo {author} {\bibfnamefont {B.}~\bibnamefont {Ide}}, \bibinfo {author} {\bibfnamefont {M.~G.}\ \bibnamefont {Gowda}}, \bibinfo {author} {\bibfnamefont {P.~J.}\ \bibnamefont {Nadkarni}},\ and\ \bibinfo {author} {\bibfnamefont {G.}~\bibnamefont {Dauphinais}},\ }\bibfield  {title} {\bibinfo {title} {{Fault-tolerant logical measurements via homological measurement}},\ }\href@noop {} {\bibfield  {journal} {\bibinfo  {journal} {arXiv preprint arXiv:2410.02753}\ } (\bibinfo {year} {2024})}\BibitemShut {NoStop}%
\bibitem [{\citenamefont {Swaroop}\ \emph {et~al.}(2024)\citenamefont {Swaroop}, \citenamefont {Jochym-O'Connor},\ and\ \citenamefont {Yoder}}]{swaroop2024universal}%
  \BibitemOpen
  \bibfield  {author} {\bibinfo {author} {\bibfnamefont {E.}~\bibnamefont {Swaroop}}, \bibinfo {author} {\bibfnamefont {T.}~\bibnamefont {Jochym-O'Connor}},\ and\ \bibinfo {author} {\bibfnamefont {T.~J.}\ \bibnamefont {Yoder}},\ }\bibfield  {title} {\bibinfo {title} {{Universal adapters between quantum {LDPC} codes}},\ }\href@noop {} {\bibfield  {journal} {\bibinfo  {journal} {arXiv preprint arXiv:2410.03628}\ } (\bibinfo {year} {2024})}\BibitemShut {NoStop}%
\bibitem [{\citenamefont {Cowtan}\ \emph {et~al.}(2025{\natexlab{a}})\citenamefont {Cowtan}, \citenamefont {He}, \citenamefont {Williamson},\ and\ \citenamefont {Yoder}}]{cowtan2025parallel}%
  \BibitemOpen
  \bibfield  {author} {\bibinfo {author} {\bibfnamefont {A.}~\bibnamefont {Cowtan}}, \bibinfo {author} {\bibfnamefont {Z.}~\bibnamefont {He}}, \bibinfo {author} {\bibfnamefont {D.~J.}\ \bibnamefont {Williamson}},\ and\ \bibinfo {author} {\bibfnamefont {T.~J.}\ \bibnamefont {Yoder}},\ }\bibfield  {title} {\bibinfo {title} {{Parallel Logical Measurements via Quantum Code Surgery}},\ }\href@noop {} {\bibfield  {journal} {\bibinfo  {journal} {arXiv preprint arXiv:2503.05003}\ } (\bibinfo {year} {2025}{\natexlab{a}})}\BibitemShut {NoStop}%
\bibitem [{\citenamefont {He}\ \emph {et~al.}(2025{\natexlab{a}})\citenamefont {He}, \citenamefont {Cowtan}, \citenamefont {Williamson},\ and\ \citenamefont {Yoder}}]{he2025extractors}%
  \BibitemOpen
  \bibfield  {author} {\bibinfo {author} {\bibfnamefont {Z.}~\bibnamefont {He}}, \bibinfo {author} {\bibfnamefont {A.}~\bibnamefont {Cowtan}}, \bibinfo {author} {\bibfnamefont {D.~J.}\ \bibnamefont {Williamson}},\ and\ \bibinfo {author} {\bibfnamefont {T.~J.}\ \bibnamefont {Yoder}},\ }\bibfield  {title} {\bibinfo {title} {Extractors: {QLDPC} architectures for efficient pauli-based computation},\ }\href@noop {} {\bibfield  {journal} {\bibinfo  {journal} {arXiv preprint arXiv:2503.10390}\ } (\bibinfo {year} {2025}{\natexlab{a}})}\BibitemShut {NoStop}%
\bibitem [{\citenamefont {Xu}\ \emph {et~al.}(2025)\citenamefont {Xu}, \citenamefont {Zhou}, \citenamefont {Bluvstein}, \citenamefont {Cain}, \citenamefont {Kalinowski}, \citenamefont {Preskill}, \citenamefont {Lukin},\ and\ \citenamefont {Maskara}}]{Xu2025Batched}%
  \BibitemOpen
  \bibfield  {author} {\bibinfo {author} {\bibfnamefont {Q.}~\bibnamefont {Xu}}, \bibinfo {author} {\bibfnamefont {H.}~\bibnamefont {Zhou}}, \bibinfo {author} {\bibfnamefont {D.}~\bibnamefont {Bluvstein}}, \bibinfo {author} {\bibfnamefont {M.}~\bibnamefont {Cain}}, \bibinfo {author} {\bibfnamefont {M.}~\bibnamefont {Kalinowski}}, \bibinfo {author} {\bibfnamefont {J.}~\bibnamefont {Preskill}}, \bibinfo {author} {\bibfnamefont {M.~D.}\ \bibnamefont {Lukin}},\ and\ \bibinfo {author} {\bibfnamefont {N.}~\bibnamefont {Maskara}},\ }\bibfield  {title} {\bibinfo {title} {{Batched high-rate logical operations for quantum LDPC codes}},\ }\href {http://arxiv.org/abs/2510.06159} {\bibfield  {journal} {\bibinfo  {journal} {arXiv preprint}\ } (\bibinfo {year} {2025})},\ \Eprint {https://arxiv.org/abs/2510.06159} {arXiv:2510.06159} \BibitemShut {NoStop}%
\bibitem [{\citenamefont {Zheng}\ \emph {et~al.}(2025)\citenamefont {Zheng}, \citenamefont {Jiang},\ and\ \citenamefont {Xu}}]{zheng2025high}%
  \BibitemOpen
  \bibfield  {author} {\bibinfo {author} {\bibfnamefont {G.}~\bibnamefont {Zheng}}, \bibinfo {author} {\bibfnamefont {L.}~\bibnamefont {Jiang}},\ and\ \bibinfo {author} {\bibfnamefont {Q.}~\bibnamefont {Xu}},\ }\bibfield  {title} {\bibinfo {title} {High-rate surgery: towards constant-overhead logical operations},\ }\href@noop {} {\bibfield  {journal} {\bibinfo  {journal} {arXiv preprint arXiv:2510.08523}\ } (\bibinfo {year} {2025})}\BibitemShut {NoStop}%
\bibitem [{\citenamefont {Baspin}\ \emph {et~al.}(2025)\citenamefont {Baspin}, \citenamefont {Berent},\ and\ \citenamefont {Cohen}}]{baspin2025fast}%
  \BibitemOpen
  \bibfield  {author} {\bibinfo {author} {\bibfnamefont {N.}~\bibnamefont {Baspin}}, \bibinfo {author} {\bibfnamefont {L.}~\bibnamefont {Berent}},\ and\ \bibinfo {author} {\bibfnamefont {L.~Z.}\ \bibnamefont {Cohen}},\ }\href {https://arxiv.org/abs/2510.04521} {\bibinfo {title} {Fast surgery for quantum ldpc codes}} (\bibinfo {year} {2025}),\ \Eprint {https://arxiv.org/abs/2510.04521} {arXiv:2510.04521 [quant-ph]} \BibitemShut {NoStop}%
\bibitem [{\citenamefont {Cowtan}\ \emph {et~al.}(2025{\natexlab{b}})\citenamefont {Cowtan}, \citenamefont {He}, \citenamefont {Williamson},\ and\ \citenamefont {Yoder}}]{Cowtan2025Fast}%
  \BibitemOpen
  \bibfield  {author} {\bibinfo {author} {\bibfnamefont {A.}~\bibnamefont {Cowtan}}, \bibinfo {author} {\bibfnamefont {Z.}~\bibnamefont {He}}, \bibinfo {author} {\bibfnamefont {D.~J.}\ \bibnamefont {Williamson}},\ and\ \bibinfo {author} {\bibfnamefont {T.~J.}\ \bibnamefont {Yoder}},\ }\bibfield  {title} {\bibinfo {title} {{Fast and fault-tolerant logical measurements: Auxiliary hypergraphs and transversal surgery}},\ }\href {https://arxiv.org/pdf/2510.14895} {\bibfield  {journal} {\bibinfo  {journal} {ArXiv preprint}\ } (\bibinfo {year} {2025}{\natexlab{b}})},\ \Eprint {https://arxiv.org/abs/2510.14895} {arXiv:2510.14895} \BibitemShut {NoStop}%
\bibitem [{\citenamefont {Bomb{\'\i}n}(2015)}]{bombin2015gauge}%
  \BibitemOpen
  \bibfield  {author} {\bibinfo {author} {\bibfnamefont {H.}~\bibnamefont {Bomb{\'\i}n}},\ }\bibfield  {title} {\bibinfo {title} {{Gauge color codes: optimal transversal gates and gauge fixing in topological stabilizer codes}},\ }\href@noop {} {\bibfield  {journal} {\bibinfo  {journal} {New Journal of Physics}\ }\textbf {\bibinfo {volume} {17}},\ \bibinfo {pages} {083002} (\bibinfo {year} {2015})}\BibitemShut {NoStop}%
\bibitem [{\citenamefont {Bomb{\'\i}n}(2016)}]{bombin2016dimensional}%
  \BibitemOpen
  \bibfield  {author} {\bibinfo {author} {\bibfnamefont {H.}~\bibnamefont {Bomb{\'\i}n}},\ }\bibfield  {title} {\bibinfo {title} {{Dimensional jump in quantum error correction}},\ }\href@noop {} {\bibfield  {journal} {\bibinfo  {journal} {New Journal of Physics}\ }\textbf {\bibinfo {volume} {18}},\ \bibinfo {pages} {043038} (\bibinfo {year} {2016})}\BibitemShut {NoStop}%
\bibitem [{\citenamefont {Vasmer}\ and\ \citenamefont {Browne}(2019)}]{Vasmer2019three}%
  \BibitemOpen
  \bibfield  {author} {\bibinfo {author} {\bibfnamefont {M.}~\bibnamefont {Vasmer}}\ and\ \bibinfo {author} {\bibfnamefont {D.~E.}\ \bibnamefont {Browne}},\ }\bibfield  {title} {\bibinfo {title} {{Three-dimensional surface codes: Transversal gates and fault-tolerant architectures}},\ }\href {https://doi.org/10.1103/physreva.100.012312} {\bibfield  {journal} {\bibinfo  {journal} {Physical Review A}\ }\textbf {\bibinfo {volume} {100}},\ \bibinfo {pages} {012312} (\bibinfo {year} {2019})}\BibitemShut {NoStop}%
\bibitem [{\citenamefont {Bombin}\ and\ \citenamefont {Martin-Delgado}(2007)}]{Bombin2007exact}%
  \BibitemOpen
  \bibfield  {author} {\bibinfo {author} {\bibfnamefont {H.}~\bibnamefont {Bombin}}\ and\ \bibinfo {author} {\bibfnamefont {M.~A.}\ \bibnamefont {Martin-Delgado}},\ }\bibfield  {title} {\bibinfo {title} {{Exact topological quantum order in {$D = 3$} and beyond: Branyons and brane-net condensates}},\ }\bibfield  {journal} {\bibinfo  {journal} {Physical Review B}\ }\textbf {\bibinfo {volume} {75}},\ \href {https://doi.org/10.1103/physrevb.75.075103} {10.1103/physrevb.75.075103} (\bibinfo {year} {2007})\BibitemShut {NoStop}%
\bibitem [{\citenamefont {Bombin}(2018)}]{Bombin20182D}%
  \BibitemOpen
  \bibfield  {author} {\bibinfo {author} {\bibfnamefont {H.}~\bibnamefont {Bombin}},\ }\bibfield  {title} {\bibinfo {title} {{2D quantum computation with 3D topological codes}},\ }\href {http://arxiv.org/abs/1810.09571} {\bibfield  {journal} {\bibinfo  {journal} {ArXiv preprint}\ } (\bibinfo {year} {2018})},\ \Eprint {https://arxiv.org/abs/1810.09571} {arXiv:1810.09571} \BibitemShut {NoStop}%
\bibitem [{\citenamefont {Brown}(2020)}]{Brown2020FTNC}%
  \BibitemOpen
  \bibfield  {author} {\bibinfo {author} {\bibfnamefont {B.~J.}\ \bibnamefont {Brown}},\ }\bibfield  {title} {\bibinfo {title} {{A fault-tolerant non-clifford gate for the surface code in two dimensions}},\ }\bibfield  {journal} {\bibinfo  {journal} {Science Advances}\ }\textbf {\bibinfo {volume} {6}},\ \href {https://doi.org/10.1126/SCIADV.AAY4929} {10.1126/SCIADV.AAY4929} (\bibinfo {year} {2020}),\ \Eprint {https://arxiv.org/abs/1903.11634} {arXiv:1903.11634} \BibitemShut {NoStop}%
\bibitem [{\citenamefont {Breuckmann}\ \emph {et~al.}(2024)\citenamefont {Breuckmann}, \citenamefont {Davydova}, \citenamefont {Eberhardt},\ and\ \citenamefont {Tantivasadakarn}}]{breuckmann2024cups}%
  \BibitemOpen
  \bibfield  {author} {\bibinfo {author} {\bibfnamefont {N.~P.}\ \bibnamefont {Breuckmann}}, \bibinfo {author} {\bibfnamefont {M.}~\bibnamefont {Davydova}}, \bibinfo {author} {\bibfnamefont {J.~N.}\ \bibnamefont {Eberhardt}},\ and\ \bibinfo {author} {\bibfnamefont {N.}~\bibnamefont {Tantivasadakarn}},\ }\bibfield  {title} {\bibinfo {title} {{Cups and Gates I: Cohomology invariants and logical quantum operations}},\ }\href@noop {} {\bibfield  {journal} {\bibinfo  {journal} {arXiv preprint arXiv:2410.16250}\ } (\bibinfo {year} {2024})}\BibitemShut {NoStop}%
\bibitem [{\citenamefont {Golowich}\ and\ \citenamefont {Lin}(2024)}]{golowich2024quantum}%
  \BibitemOpen
  \bibfield  {author} {\bibinfo {author} {\bibfnamefont {L.}~\bibnamefont {Golowich}}\ and\ \bibinfo {author} {\bibfnamefont {T.-C.}\ \bibnamefont {Lin}},\ }\bibfield  {title} {\bibinfo {title} {{Quantum {LDPC} Codes with Transversal Non-Clifford Gates via Products of Algebraic Codes}},\ }\href@noop {} {\bibfield  {journal} {\bibinfo  {journal} {arXiv preprint arXiv:2410.14662}\ } (\bibinfo {year} {2024})}\BibitemShut {NoStop}%
\bibitem [{\citenamefont {Scruby}\ \emph {et~al.}(2024)\citenamefont {Scruby}, \citenamefont {Pesah},\ and\ \citenamefont {Webster}}]{scruby2024quantum}%
  \BibitemOpen
  \bibfield  {author} {\bibinfo {author} {\bibfnamefont {T.~R.}\ \bibnamefont {Scruby}}, \bibinfo {author} {\bibfnamefont {A.}~\bibnamefont {Pesah}},\ and\ \bibinfo {author} {\bibfnamefont {M.}~\bibnamefont {Webster}},\ }\bibfield  {title} {\bibinfo {title} {{Quantum rainbow codes}},\ }\href@noop {} {\bibfield  {journal} {\bibinfo  {journal} {arXiv preprint arXiv:2408.13130}\ } (\bibinfo {year} {2024})}\BibitemShut {NoStop}%
\bibitem [{\citenamefont {Hsin}\ \emph {et~al.}(2024)\citenamefont {Hsin}, \citenamefont {Kobayashi},\ and\ \citenamefont {Zhu}}]{Hsin2024}%
  \BibitemOpen
  \bibfield  {author} {\bibinfo {author} {\bibfnamefont {P.-S.}\ \bibnamefont {Hsin}}, \bibinfo {author} {\bibfnamefont {R.}~\bibnamefont {Kobayashi}},\ and\ \bibinfo {author} {\bibfnamefont {G.}~\bibnamefont {Zhu}},\ }\bibfield  {title} {\bibinfo {title} {{Classifying Logical Gates in Quantum Codes via Cohomology Operations and Symmetry}},\ }\href {http://arxiv.org/abs/2411.15848} {\bibfield  {journal} {\bibinfo  {journal} {ArXiv preprint}\ } (\bibinfo {year} {2024})},\ \Eprint {https://arxiv.org/abs/2411.15848} {arXiv:2411.15848} \BibitemShut {NoStop}%
\bibitem [{\citenamefont {Zhu}(2025{\natexlab{a}})}]{Zhu2025b}%
  \BibitemOpen
  \bibfield  {author} {\bibinfo {author} {\bibfnamefont {G.}~\bibnamefont {Zhu}},\ }\bibfield  {title} {\bibinfo {title} {{Transversal non-Clifford gates on qLDPC codes breaking the $\sqrt{N}$ distance barrier and quantum-inspired geometry with $\mathbb{Z}_2$ systolic freedom}},\ }\href {http://arxiv.org/abs/2507.15056} {\bibfield  {journal} {\bibinfo  {journal} {ArXiv preprint}\ } (\bibinfo {year} {2025}{\natexlab{a}})},\ \Eprint {https://arxiv.org/abs/2507.15056} {arXiv:2507.15056} \BibitemShut {NoStop}%
\bibitem [{\citenamefont {Zhu}\ \emph {et~al.}(2023{\natexlab{a}})\citenamefont {Zhu}, \citenamefont {Sikander}, \citenamefont {Portnoy}, \citenamefont {Cross},\ and\ \citenamefont {Brown}}]{zhu2023gates}%
  \BibitemOpen
  \bibfield  {author} {\bibinfo {author} {\bibfnamefont {G.}~\bibnamefont {Zhu}}, \bibinfo {author} {\bibfnamefont {S.}~\bibnamefont {Sikander}}, \bibinfo {author} {\bibfnamefont {E.}~\bibnamefont {Portnoy}}, \bibinfo {author} {\bibfnamefont {A.~W.}\ \bibnamefont {Cross}},\ and\ \bibinfo {author} {\bibfnamefont {B.~J.}\ \bibnamefont {Brown}},\ }\bibfield  {title} {\bibinfo {title} {{Non-Clifford and parallelizable fault-tolerant logical gates on constant and almost-constant rate homological quantum {LDPC} codes via higher symmetries}},\ }\href@noop {} {\bibfield  {journal} {\bibinfo  {journal} {Arxiv preprint}\ } (\bibinfo {year} {2023}{\natexlab{a}})}\BibitemShut {NoStop}%
\bibitem [{\citenamefont {Zhu}(2025{\natexlab{b}})}]{zhu2025topological}%
  \BibitemOpen
  \bibfield  {author} {\bibinfo {author} {\bibfnamefont {G.}~\bibnamefont {Zhu}},\ }\bibfield  {title} {\bibinfo {title} {{A topological theory for qLDPC: non-Clifford gates and magic state fountain on homological product codes with constant rate and beyond the $ N^{1/3} $ distance barrier}},\ }\href@noop {} {\bibfield  {journal} {\bibinfo  {journal} {arXiv preprint arXiv:2501.19375}\ } (\bibinfo {year} {2025}{\natexlab{b}})}\BibitemShut {NoStop}%
\bibitem [{\citenamefont {Gulshen}\ and\ \citenamefont {Kaufman}(2025)}]{Gulshen2025}%
  \BibitemOpen
  \bibfield  {author} {\bibinfo {author} {\bibfnamefont {K.}~\bibnamefont {Gulshen}}\ and\ \bibinfo {author} {\bibfnamefont {T.}~\bibnamefont {Kaufman}},\ }\bibfield  {title} {\bibinfo {title} {{Quantum Tanner Color Codes on Qubits with Transversal Gates}},\ }\href {https://arxiv.org/pdf/2510.07864} {\bibfield  {journal} {\bibinfo  {journal} {arXiv:2510.07864}\ } (\bibinfo {year} {2025})},\ \Eprint {https://arxiv.org/abs/2510.07864} {arXiv:2510.07864} \BibitemShut {NoStop}%
\bibitem [{\citenamefont {Golowich}\ and\ \citenamefont {Guruswami}(2025)}]{Golowich2025}%
  \BibitemOpen
  \bibfield  {author} {\bibinfo {author} {\bibfnamefont {L.}~\bibnamefont {Golowich}}\ and\ \bibinfo {author} {\bibfnamefont {V.}~\bibnamefont {Guruswami}},\ }\bibfield  {title} {\bibinfo {title} {{Near-Asymptotically-Good Quantum Codes with Transversal CCZ Gates and Sublinear-Weight Parity-Checks}},\ }\href {https://arxiv.org/pdf/2510.06798} {\bibfield  {journal} {\bibinfo  {journal} {ArXiv preprint}\ } (\bibinfo {year} {2025})},\ \Eprint {https://arxiv.org/abs/2510.06798} {arXiv:2510.06798} \BibitemShut {NoStop}%
\bibitem [{\citenamefont {He}\ \emph {et~al.}(2025{\natexlab{b}})\citenamefont {He}, \citenamefont {Vaikuntanathan}, \citenamefont {Wills},\ and\ \citenamefont {Zhang}}]{He2025}%
  \BibitemOpen
  \bibfield  {author} {\bibinfo {author} {\bibfnamefont {Z.}~\bibnamefont {He}}, \bibinfo {author} {\bibfnamefont {V.}~\bibnamefont {Vaikuntanathan}}, \bibinfo {author} {\bibfnamefont {A.}~\bibnamefont {Wills}},\ and\ \bibinfo {author} {\bibfnamefont {R.~Y.}\ \bibnamefont {Zhang}},\ }\bibfield  {title} {\bibinfo {title} {{Asymptotically Good Quantum Codes with Addressable and Transversal Non-Clifford Gates}},\ }\href {https://arxiv.org/pdf/2507.05392} {\bibfield  {journal} {\bibinfo  {journal} {ArXiv preprint}\ } (\bibinfo {year} {2025}{\natexlab{b}})},\ \Eprint {https://arxiv.org/abs/2507.05392} {arXiv:2507.05392} \BibitemShut {NoStop}%
\bibitem [{\citenamefont {Golowich}\ \emph {et~al.}(2025)\citenamefont {Golowich}, \citenamefont {Chang},\ and\ \citenamefont {Zhu}}]{golowich2025constant}%
  \BibitemOpen
  \bibfield  {author} {\bibinfo {author} {\bibfnamefont {L.}~\bibnamefont {Golowich}}, \bibinfo {author} {\bibfnamefont {K.}~\bibnamefont {Chang}},\ and\ \bibinfo {author} {\bibfnamefont {G.}~\bibnamefont {Zhu}},\ }\bibfield  {title} {\bibinfo {title} {Constant-overhead addressable gates via single-shot code switching},\ }\href@noop {} {\bibfield  {journal} {\bibinfo  {journal} {arXiv preprint arXiv:2510.06760}\ } (\bibinfo {year} {2025})}\BibitemShut {NoStop}%
\bibitem [{\citenamefont {Tan}\ \emph {et~al.}(2025)\citenamefont {Tan}, \citenamefont {Hong}, \citenamefont {Lin}, \citenamefont {Gullans},\ and\ \citenamefont {Hsieh}}]{Tan2025Single}%
  \BibitemOpen
  \bibfield  {author} {\bibinfo {author} {\bibfnamefont {S.~J.~S.}\ \bibnamefont {Tan}}, \bibinfo {author} {\bibfnamefont {Y.}~\bibnamefont {Hong}}, \bibinfo {author} {\bibfnamefont {T.-C.}\ \bibnamefont {Lin}}, \bibinfo {author} {\bibfnamefont {M.~J.}\ \bibnamefont {Gullans}},\ and\ \bibinfo {author} {\bibfnamefont {M.-H.}\ \bibnamefont {Hsieh}},\ }\bibfield  {title} {\bibinfo {title} {{Single-Shot Universality in Quantum LDPC Codes via Code-Switching}},\ }\href {https://arxiv.org/pdf/2510.08552} {\bibfield  {journal} {\bibinfo  {journal} {arXiv:2510.08552}\ } (\bibinfo {year} {2025})},\ \Eprint {https://arxiv.org/abs/2510.08552} {arXiv:2510.08552} \BibitemShut {NoStop}%
\bibitem [{\citenamefont {Jacob}\ \emph {et~al.}(2025)\citenamefont {Jacob}, \citenamefont {McLauchlan},\ and\ \citenamefont {Browne}}]{Jacob2025}%
  \BibitemOpen
  \bibfield  {author} {\bibinfo {author} {\bibfnamefont {A.}~\bibnamefont {Jacob}}, \bibinfo {author} {\bibfnamefont {C.}~\bibnamefont {McLauchlan}},\ and\ \bibinfo {author} {\bibfnamefont {D.~E.}\ \bibnamefont {Browne}},\ }\bibfield  {title} {\bibinfo {title} {{Single-Shot Decoding and Fault-tolerant Gates with Trivariate Tricycle Codes}},\ }\href {http://arxiv.org/abs/2508.08191} {\bibfield  {journal} {\bibinfo  {journal} {ArXiv preprint}\ } (\bibinfo {year} {2025})},\ \Eprint {https://arxiv.org/abs/2508.08191} {arXiv:2508.08191} \BibitemShut {NoStop}%
\bibitem [{\citenamefont {Menon}\ \emph {et~al.}(2025)\citenamefont {Menon}, \citenamefont {Bonilla-Ataides}, \citenamefont {Mehta}, \citenamefont {Gu}, \citenamefont {Tan},\ and\ \citenamefont {Lukin}}]{Menon2025}%
  \BibitemOpen
  \bibfield  {author} {\bibinfo {author} {\bibfnamefont {V.}~\bibnamefont {Menon}}, \bibinfo {author} {\bibfnamefont {J.~P.}\ \bibnamefont {Bonilla-Ataides}}, \bibinfo {author} {\bibfnamefont {R.}~\bibnamefont {Mehta}}, \bibinfo {author} {\bibfnamefont {A.}~\bibnamefont {Gu}}, \bibinfo {author} {\bibfnamefont {D.~B.}\ \bibnamefont {Tan}},\ and\ \bibinfo {author} {\bibfnamefont {M.~D.}\ \bibnamefont {Lukin}},\ }\bibfield  {title} {\bibinfo {title} {{Magic tricycles: Efficient magic state generation with finite block-length quantum LDPC codes}},\ }\href {http://arxiv.org/abs/2508.10714} {\bibfield  {journal} {\bibinfo  {journal} {ArXiv preprint}\ } (\bibinfo {year} {2025})},\ \Eprint {https://arxiv.org/abs/2508.10714} {arXiv:2508.10714} \BibitemShut {NoStop}%
\bibitem [{\citenamefont {Kobayashi}\ \emph {et~al.}(2025)\citenamefont {Kobayashi}, \citenamefont {Zhu},\ and\ \citenamefont {Hsin}}]{Kobayashi2025}%
  \BibitemOpen
  \bibfield  {author} {\bibinfo {author} {\bibfnamefont {R.}~\bibnamefont {Kobayashi}}, \bibinfo {author} {\bibfnamefont {G.}~\bibnamefont {Zhu}},\ and\ \bibinfo {author} {\bibfnamefont {P.-S.}\ \bibnamefont {Hsin}},\ }\bibfield  {title} {\bibinfo {title} {{Clifford Hierarchy Stabilizer Codes: Transversal Non-Clifford Gates and Magic}},\ }\href {http://arxiv.org/abs/2511.02900} {\bibfield  {journal} {\bibinfo  {journal} {ArXiv preprint}\ } (\bibinfo {year} {2025})},\ \Eprint {https://arxiv.org/abs/2511.02900} {arXiv:2511.02900} \BibitemShut {NoStop}%
\bibitem [{\citenamefont {Bravyi}\ and\ \citenamefont {Kitaev}(2005)}]{bravyi2005magic}%
  \BibitemOpen
  \bibfield  {author} {\bibinfo {author} {\bibfnamefont {S.}~\bibnamefont {Bravyi}}\ and\ \bibinfo {author} {\bibfnamefont {A.}~\bibnamefont {Kitaev}},\ }\bibfield  {title} {\bibinfo {title} {{Universal quantum computation with ideal Clifford gates and noisy ancillas}},\ }\href {https://doi.org/10.1103/PhysRevA.71.022316} {\bibfield  {journal} {\bibinfo  {journal} {Phys. Rev. A}\ }\textbf {\bibinfo {volume} {71}},\ \bibinfo {pages} {022316} (\bibinfo {year} {2005})}\BibitemShut {NoStop}%
\bibitem [{\citenamefont {Meier}\ \emph {et~al.}(2013)\citenamefont {Meier}, \citenamefont {Eastin},\ and\ \citenamefont {Knill}}]{meier2012magic}%
  \BibitemOpen
  \bibfield  {author} {\bibinfo {author} {\bibfnamefont {A.~M.}\ \bibnamefont {Meier}}, \bibinfo {author} {\bibfnamefont {B.}~\bibnamefont {Eastin}},\ and\ \bibinfo {author} {\bibfnamefont {E.}~\bibnamefont {Knill}},\ }\bibfield  {title} {\bibinfo {title} {{Magic-state distillation with the four-qubit code}},\ }\href {https://www.rintonpress.com/journals/doi/QIC13.3-4-2.html} {\bibfield  {journal} {\bibinfo  {journal} {Quantum Information and Computation}\ }\textbf {\bibinfo {volume} {13}},\ \bibinfo {pages} {195–209} (\bibinfo {year} {2013})}\BibitemShut {NoStop}%
\bibitem [{\citenamefont {Bravyi}\ and\ \citenamefont {Haah}(2012)}]{Bravyi2012magic}%
  \BibitemOpen
  \bibfield  {author} {\bibinfo {author} {\bibfnamefont {S.}~\bibnamefont {Bravyi}}\ and\ \bibinfo {author} {\bibfnamefont {J.}~\bibnamefont {Haah}},\ }\bibfield  {title} {\bibinfo {title} {{Magic-state distillation with low overhead}},\ }\href {https://doi.org/10.1103/physreva.86.052329} {\bibfield  {journal} {\bibinfo  {journal} {Physical Review A}\ }\textbf {\bibinfo {volume} {86}},\ \bibinfo {pages} {052329} (\bibinfo {year} {2012})}\BibitemShut {NoStop}%
\bibitem [{\citenamefont {Campbell}\ \emph {et~al.}(2012)\citenamefont {Campbell}, \citenamefont {Anwar},\ and\ \citenamefont {Browne}}]{campbell2012magic}%
  \BibitemOpen
  \bibfield  {author} {\bibinfo {author} {\bibfnamefont {E.~T.}\ \bibnamefont {Campbell}}, \bibinfo {author} {\bibfnamefont {H.}~\bibnamefont {Anwar}},\ and\ \bibinfo {author} {\bibfnamefont {D.~E.}\ \bibnamefont {Browne}},\ }\bibfield  {title} {\bibinfo {title} {{Magic-state distillation in all prime dimensions using quantum reed-muller codes}},\ }\href {https://journals.aps.org/prx/abstract/10.1103/PhysRevX.2.041021} {\bibfield  {journal} {\bibinfo  {journal} {Physical Review X}\ }\textbf {\bibinfo {volume} {2}},\ \bibinfo {pages} {041021} (\bibinfo {year} {2012})}\BibitemShut {NoStop}%
\bibitem [{\citenamefont {Jones}(2013)}]{Jones2013multilevel}%
  \BibitemOpen
  \bibfield  {author} {\bibinfo {author} {\bibfnamefont {C.}~\bibnamefont {Jones}},\ }\bibfield  {title} {\bibinfo {title} {{Multilevel distillation of magic states for quantum computing}},\ }\href {https://doi.org/10.1103/physreva.87.042305} {\bibfield  {journal} {\bibinfo  {journal} {Physical Review A}\ }\textbf {\bibinfo {volume} {87}},\ \bibinfo {pages} {042305} (\bibinfo {year} {2013})}\BibitemShut {NoStop}%
\bibitem [{\citenamefont {Haah}\ \emph {et~al.}(2017)\citenamefont {Haah}, \citenamefont {Hastings}, \citenamefont {Poulin},\ and\ \citenamefont {Wecker}}]{haah2017magic}%
  \BibitemOpen
  \bibfield  {author} {\bibinfo {author} {\bibfnamefont {J.}~\bibnamefont {Haah}}, \bibinfo {author} {\bibfnamefont {M.~B.}\ \bibnamefont {Hastings}}, \bibinfo {author} {\bibfnamefont {D.}~\bibnamefont {Poulin}},\ and\ \bibinfo {author} {\bibfnamefont {D.}~\bibnamefont {Wecker}},\ }\bibfield  {title} {\bibinfo {title} {{Magic state distillation with low space overhead and optimal asymptotic input count}},\ }\href@noop {} {\bibfield  {journal} {\bibinfo  {journal} {Quantum}\ }\textbf {\bibinfo {volume} {1}},\ \bibinfo {pages} {31} (\bibinfo {year} {2017})}\BibitemShut {NoStop}%
\bibitem [{\citenamefont {Krishna}\ and\ \citenamefont {Tillich}(2018)}]{krishna2018magic}%
  \BibitemOpen
  \bibfield  {author} {\bibinfo {author} {\bibfnamefont {A.}~\bibnamefont {Krishna}}\ and\ \bibinfo {author} {\bibfnamefont {J.-P.}\ \bibnamefont {Tillich}},\ }\bibfield  {title} {\bibinfo {title} {{Magic state distillation with punctured polar codes}},\ }\href {https://arxiv.org/abs/1811.03112} {\bibfield  {journal} {\bibinfo  {journal} {arXiv preprint arXiv:1811.03112}\ } (\bibinfo {year} {2018})}\BibitemShut {NoStop}%
\bibitem [{\citenamefont {Krishna}\ and\ \citenamefont {Tillich}(2019)}]{Krishna2019towards}%
  \BibitemOpen
  \bibfield  {author} {\bibinfo {author} {\bibfnamefont {A.}~\bibnamefont {Krishna}}\ and\ \bibinfo {author} {\bibfnamefont {J.-P.}\ \bibnamefont {Tillich}},\ }\bibfield  {title} {\bibinfo {title} {{Towards Low Overhead Magic State Distillation}},\ }\href {https://doi.org/10.1103/physrevlett.123.070507} {\bibfield  {journal} {\bibinfo  {journal} {Physical Review Letters}\ }\textbf {\bibinfo {volume} {123}},\ \bibinfo {pages} {070507} (\bibinfo {year} {2019})}\BibitemShut {NoStop}%
\bibitem [{\citenamefont {Litinski}(2019)}]{litinski2019magic}%
  \BibitemOpen
  \bibfield  {author} {\bibinfo {author} {\bibfnamefont {D.}~\bibnamefont {Litinski}},\ }\bibfield  {title} {\bibinfo {title} {{Magic state distillation: Not as costly as you think}},\ }\href {https://quantum-journal.org/papers/q-2019-12-02-205/} {\bibfield  {journal} {\bibinfo  {journal} {Quantum}\ }\textbf {\bibinfo {volume} {3}},\ \bibinfo {pages} {205} (\bibinfo {year} {2019})}\BibitemShut {NoStop}%
\bibitem [{\citenamefont {Rodriguez}\ \emph {et~al.}(2024)\citenamefont {Rodriguez}, \citenamefont {Robinson}, \citenamefont {Jepsen}, \citenamefont {He}, \citenamefont {Duckering}, \citenamefont {Zhao}, \citenamefont {Wu}, \citenamefont {Campo}, \citenamefont {Bagnall}, \citenamefont {Kwon} \emph {et~al.}}]{rodriguez2024experimental}%
  \BibitemOpen
  \bibfield  {author} {\bibinfo {author} {\bibfnamefont {P.~S.}\ \bibnamefont {Rodriguez}}, \bibinfo {author} {\bibfnamefont {J.~M.}\ \bibnamefont {Robinson}}, \bibinfo {author} {\bibfnamefont {P.~N.}\ \bibnamefont {Jepsen}}, \bibinfo {author} {\bibfnamefont {Z.}~\bibnamefont {He}}, \bibinfo {author} {\bibfnamefont {C.}~\bibnamefont {Duckering}}, \bibinfo {author} {\bibfnamefont {C.}~\bibnamefont {Zhao}}, \bibinfo {author} {\bibfnamefont {K.-H.}\ \bibnamefont {Wu}}, \bibinfo {author} {\bibfnamefont {J.}~\bibnamefont {Campo}}, \bibinfo {author} {\bibfnamefont {K.}~\bibnamefont {Bagnall}}, \bibinfo {author} {\bibfnamefont {M.}~\bibnamefont {Kwon}}, \emph {et~al.},\ }\bibfield  {title} {\bibinfo {title} {{Experimental Demonstration of Logical Magic State Distillation}},\ }\href@noop {} {\bibfield  {journal} {\bibinfo  {journal} {arXiv preprint arXiv:2412.15165}\ } (\bibinfo {year} {2024})}\BibitemShut {NoStop}%
\bibitem [{\citenamefont {Wills}\ \emph {et~al.}(2024)\citenamefont {Wills}, \citenamefont {Hsieh},\ and\ \citenamefont {Yamasaki}}]{wills2024constant}%
  \BibitemOpen
  \bibfield  {author} {\bibinfo {author} {\bibfnamefont {A.}~\bibnamefont {Wills}}, \bibinfo {author} {\bibfnamefont {M.-H.}\ \bibnamefont {Hsieh}},\ and\ \bibinfo {author} {\bibfnamefont {H.}~\bibnamefont {Yamasaki}},\ }\bibfield  {title} {\bibinfo {title} {{Constant-overhead magic state distillation}},\ }\href@noop {} {\bibfield  {journal} {\bibinfo  {journal} {arXiv preprint arXiv:2408.07764}\ } (\bibinfo {year} {2024})}\BibitemShut {NoStop}%
\bibitem [{\citenamefont {Nguyen}(2024)}]{Nguyen2024}%
  \BibitemOpen
  \bibfield  {author} {\bibinfo {author} {\bibfnamefont {Q.~T.}\ \bibnamefont {Nguyen}},\ }\bibfield  {title} {\bibinfo {title} {{Good binary quantum codes with transversal CCZ gate}},\ }\href {https://doi.org/10.1145/3717823.3718186} {\bibfield  {journal} {\bibinfo  {journal} {Proceedings of the Annual ACM Symposium on Theory of Computing}\ ,\ \bibinfo {pages} {697}} (\bibinfo {year} {2024})},\ \Eprint {https://arxiv.org/abs/2408.10140} {arXiv:2408.10140} \BibitemShut {NoStop}%
\bibitem [{\citenamefont {Li}(2015)}]{li2015magic}%
  \BibitemOpen
  \bibfield  {author} {\bibinfo {author} {\bibfnamefont {Y.}~\bibnamefont {Li}},\ }\bibfield  {title} {\bibinfo {title} {{A magic state’s fidelity can be superior to the operations that created it}},\ }\href {https://doi.org/10.1088/1367-2630/17/2/023037} {\bibfield  {journal} {\bibinfo  {journal} {New Journal of Physics}\ }\textbf {\bibinfo {volume} {17}},\ \bibinfo {pages} {023037} (\bibinfo {year} {2015})}\BibitemShut {NoStop}%
\bibitem [{\citenamefont {Chamberland}\ and\ \citenamefont {Cross}(2019)}]{Chamberland2019faulttolerantmagic}%
  \BibitemOpen
  \bibfield  {author} {\bibinfo {author} {\bibfnamefont {C.}~\bibnamefont {Chamberland}}\ and\ \bibinfo {author} {\bibfnamefont {A.~W.}\ \bibnamefont {Cross}},\ }\bibfield  {title} {\bibinfo {title} {{Fault-tolerant magic state preparation with flag qubits}},\ }\href {https://doi.org/10.22331/q-2019-05-20-143} {\bibfield  {journal} {\bibinfo  {journal} {{Quantum}}\ }\textbf {\bibinfo {volume} {3}},\ \bibinfo {pages} {143} (\bibinfo {year} {2019})}\BibitemShut {NoStop}%
\bibitem [{\citenamefont {Chamberland}\ and\ \citenamefont {Noh}(2020)}]{Chamberland2020verylow}%
  \BibitemOpen
  \bibfield  {author} {\bibinfo {author} {\bibfnamefont {C.}~\bibnamefont {Chamberland}}\ and\ \bibinfo {author} {\bibfnamefont {K.}~\bibnamefont {Noh}},\ }\bibfield  {title} {\bibinfo {title} {{Very low overhead fault-tolerant magic state preparation using redundant ancilla encoding and flag qubits}},\ }\bibfield  {journal} {\bibinfo  {journal} {npj Quantum Information}\ }\textbf {\bibinfo {volume} {6}},\ \href {https://doi.org/10.1038/s41534-020-00319-5} {10.1038/s41534-020-00319-5} (\bibinfo {year} {2020})\BibitemShut {NoStop}%
\bibitem [{\citenamefont {Gidney}\ \emph {et~al.}(2024)\citenamefont {Gidney}, \citenamefont {Shutty},\ and\ \citenamefont {Jones}}]{gidney2024magic}%
  \BibitemOpen
  \bibfield  {author} {\bibinfo {author} {\bibfnamefont {C.}~\bibnamefont {Gidney}}, \bibinfo {author} {\bibfnamefont {N.}~\bibnamefont {Shutty}},\ and\ \bibinfo {author} {\bibfnamefont {C.}~\bibnamefont {Jones}},\ }\bibfield  {title} {\bibinfo {title} {{Magic state cultivation: growing T states as cheap as CNOT gates}},\ }\href@noop {} {\bibfield  {journal} {\bibinfo  {journal} {arXiv preprint arXiv:2409.17595}\ } (\bibinfo {year} {2024})}\BibitemShut {NoStop}%
\bibitem [{\citenamefont {Gupta}\ \emph {et~al.}(2024)\citenamefont {Gupta}, \citenamefont {Sundaresan}, \citenamefont {Alexander}, \citenamefont {Wood}, \citenamefont {Merkel}, \citenamefont {Healy}, \citenamefont {Hillenbrand}, \citenamefont {Jochym-O’Connor}, \citenamefont {Wootton}, \citenamefont {Yoder} \emph {et~al.}}]{gupta2024encoding}%
  \BibitemOpen
  \bibfield  {author} {\bibinfo {author} {\bibfnamefont {R.~S.}\ \bibnamefont {Gupta}}, \bibinfo {author} {\bibfnamefont {N.}~\bibnamefont {Sundaresan}}, \bibinfo {author} {\bibfnamefont {T.}~\bibnamefont {Alexander}}, \bibinfo {author} {\bibfnamefont {C.~J.}\ \bibnamefont {Wood}}, \bibinfo {author} {\bibfnamefont {S.~T.}\ \bibnamefont {Merkel}}, \bibinfo {author} {\bibfnamefont {M.~B.}\ \bibnamefont {Healy}}, \bibinfo {author} {\bibfnamefont {M.}~\bibnamefont {Hillenbrand}}, \bibinfo {author} {\bibfnamefont {T.}~\bibnamefont {Jochym-O’Connor}}, \bibinfo {author} {\bibfnamefont {J.~R.}\ \bibnamefont {Wootton}}, \bibinfo {author} {\bibfnamefont {T.~J.}\ \bibnamefont {Yoder}}, \emph {et~al.},\ }\bibfield  {title} {\bibinfo {title} {{Encoding a magic state with beyond break-even fidelity}},\ }\href {https://www.nature.com/articles/s41586-023-06846-3} {\bibfield  {journal} {\bibinfo  {journal} {Nature}\ }\textbf {\bibinfo {volume} {625}},\ \bibinfo {pages} {259} (\bibinfo {year} {2024})}\BibitemShut {NoStop}%
\bibitem [{\citenamefont {Hirano}\ \emph {et~al.}(2024)\citenamefont {Hirano}, \citenamefont {Itogawa},\ and\ \citenamefont {Fujii}}]{Hirano2024zerolevel}%
  \BibitemOpen
  \bibfield  {author} {\bibinfo {author} {\bibfnamefont {Y.}~\bibnamefont {Hirano}}, \bibinfo {author} {\bibfnamefont {T.}~\bibnamefont {Itogawa}},\ and\ \bibinfo {author} {\bibfnamefont {K.}~\bibnamefont {Fujii}},\ }\href {https://doi.org/10.48550/ARXIV.2404.09740} {\bibinfo {title} {{Leveraging Zero-Level Distillation to Generate High-Fidelity Magic States}}} (\bibinfo {year} {2024})\BibitemShut {NoStop}%
\bibitem [{\citenamefont {Itogawa}\ \emph {et~al.}(2025)\citenamefont {Itogawa}, \citenamefont {Takada}, \citenamefont {Hirano},\ and\ \citenamefont {Fujii}}]{Itogawa2025}%
  \BibitemOpen
  \bibfield  {author} {\bibinfo {author} {\bibfnamefont {T.}~\bibnamefont {Itogawa}}, \bibinfo {author} {\bibfnamefont {Y.}~\bibnamefont {Takada}}, \bibinfo {author} {\bibfnamefont {Y.}~\bibnamefont {Hirano}},\ and\ \bibinfo {author} {\bibfnamefont {K.}~\bibnamefont {Fujii}},\ }\bibfield  {title} {\bibinfo {title} {{Efficient Magic State Distillation by Zero-Level Distillation}},\ }\bibfield  {journal} {\bibinfo  {journal} {PRX Quantum}\ }\textbf {\bibinfo {volume} {6}},\ \href {https://doi.org/10.1103/thxx-njr6} {10.1103/thxx-njr6} (\bibinfo {year} {2025}),\ \Eprint {https://arxiv.org/abs/2403.03991v2} {arXiv:2403.03991v2} \BibitemShut {NoStop}%
\bibitem [{\citenamefont {Chen}\ \emph {et~al.}(2025)\citenamefont {Chen}, \citenamefont {Chen}, \citenamefont {Lu},\ and\ \citenamefont {Pan}}]{Chen2025}%
  \BibitemOpen
  \bibfield  {author} {\bibinfo {author} {\bibfnamefont {Z.-H.}\ \bibnamefont {Chen}}, \bibinfo {author} {\bibfnamefont {M.-C.}\ \bibnamefont {Chen}}, \bibinfo {author} {\bibfnamefont {C.-Y.}\ \bibnamefont {Lu}},\ and\ \bibinfo {author} {\bibfnamefont {J.-W.}\ \bibnamefont {Pan}},\ }\bibfield  {title} {\bibinfo {title} {{Efficient Magic State Cultivation on $\mathbb{RP}^2$}},\ }\href {http://arxiv.org/abs/2503.18657} {\bibfield  {journal} {\bibinfo  {journal} {Arxiv preprint}\ } (\bibinfo {year} {2025})},\ \Eprint {https://arxiv.org/abs/2503.18657} {arXiv:2503.18657} \BibitemShut {NoStop}%
\bibitem [{\citenamefont {Vaknin}\ \emph {et~al.}(2025)\citenamefont {Vaknin}, \citenamefont {Jacoby}, \citenamefont {Grimsmo},\ and\ \citenamefont {Retzker}}]{Vaknin2025}%
  \BibitemOpen
  \bibfield  {author} {\bibinfo {author} {\bibfnamefont {Y.}~\bibnamefont {Vaknin}}, \bibinfo {author} {\bibfnamefont {S.}~\bibnamefont {Jacoby}}, \bibinfo {author} {\bibfnamefont {A.}~\bibnamefont {Grimsmo}},\ and\ \bibinfo {author} {\bibfnamefont {A.}~\bibnamefont {Retzker}},\ }\bibfield  {title} {\bibinfo {title} {{Efficient Magic State Cultivation on the Surface Code}},\ }\href {http://arxiv.org/abs/2502.01743} {\bibfield  {journal} {\bibinfo  {journal} {Arxiv preprint}\ } (\bibinfo {year} {2025})},\ \Eprint {https://arxiv.org/abs/2502.01743} {arXiv:2502.01743} \BibitemShut {NoStop}%
\bibitem [{\citenamefont {Sahay}\ \emph {et~al.}(2025)\citenamefont {Sahay}, \citenamefont {Tsai}, \citenamefont {Chang}, \citenamefont {Su}, \citenamefont {Smith}, \citenamefont {Singh},\ and\ \citenamefont {Puri}}]{Sahay2025}%
  \BibitemOpen
  \bibfield  {author} {\bibinfo {author} {\bibfnamefont {K.}~\bibnamefont {Sahay}}, \bibinfo {author} {\bibfnamefont {P.-K.}\ \bibnamefont {Tsai}}, \bibinfo {author} {\bibfnamefont {K.}~\bibnamefont {Chang}}, \bibinfo {author} {\bibfnamefont {Q.}~\bibnamefont {Su}}, \bibinfo {author} {\bibfnamefont {T.~B.}\ \bibnamefont {Smith}}, \bibinfo {author} {\bibfnamefont {S.}~\bibnamefont {Singh}},\ and\ \bibinfo {author} {\bibfnamefont {S.}~\bibnamefont {Puri}},\ }\bibfield  {title} {\bibinfo {title} {{Fold-transversal surface code cultivation}},\ }\href {http://arxiv.org/abs/2509.05212} {\bibfield  {journal} {\bibinfo  {journal} {Arxiv preprint}\ } (\bibinfo {year} {2025})},\ \Eprint {https://arxiv.org/abs/2509.05212} {arXiv:2509.05212} \BibitemShut {NoStop}%
\bibitem [{\citenamefont {Claes}(2025)}]{Claes2025}%
  \BibitemOpen
  \bibfield  {author} {\bibinfo {author} {\bibfnamefont {J.}~\bibnamefont {Claes}},\ }\bibfield  {title} {\bibinfo {title} {{Cultivating T states on the surface code with only two-qubit gates}},\ }\href {http://arxiv.org/abs/2509.05232} {\bibfield  {journal} {\bibinfo  {journal} {Arxiv preprint}\ } (\bibinfo {year} {2025})},\ \Eprint {https://arxiv.org/abs/2509.05232} {arXiv:2509.05232} \BibitemShut {NoStop}%
\bibitem [{\citenamefont {Davydova}\ \emph {et~al.}(2025)\citenamefont {Davydova}, \citenamefont {Bauer}, \citenamefont {de~la Fuente}, \citenamefont {Webster}, \citenamefont {Williamson},\ and\ \citenamefont {Brown}}]{davydova2025universal}%
  \BibitemOpen
  \bibfield  {author} {\bibinfo {author} {\bibfnamefont {M.}~\bibnamefont {Davydova}}, \bibinfo {author} {\bibfnamefont {A.}~\bibnamefont {Bauer}}, \bibinfo {author} {\bibfnamefont {J.~C.~M.}\ \bibnamefont {de~la Fuente}}, \bibinfo {author} {\bibfnamefont {M.}~\bibnamefont {Webster}}, \bibinfo {author} {\bibfnamefont {D.~J.}\ \bibnamefont {Williamson}},\ and\ \bibinfo {author} {\bibfnamefont {B.~J.}\ \bibnamefont {Brown}},\ }\bibfield  {title} {\bibinfo {title} {Universal fault tolerant quantum computation in 2d without getting tied in knots},\ }\href@noop {} {\bibfield  {journal} {\bibinfo  {journal} {arXiv preprint arxiv:2503.15751}\ } (\bibinfo {year} {2025})}\BibitemShut {NoStop}%
\bibitem [{\citenamefont {Huang}\ and\ \citenamefont {Chen}(2025)}]{Huang2025}%
  \BibitemOpen
  \bibfield  {author} {\bibinfo {author} {\bibfnamefont {S.-J.}\ \bibnamefont {Huang}}\ and\ \bibinfo {author} {\bibfnamefont {Y.}~\bibnamefont {Chen}},\ }\bibfield  {title} {\bibinfo {title} {{Generating logical magic states with the aid of non-Abelian topological order}},\ }\href {http://arxiv.org/abs/2502.00998} {\bibfield  {journal} {\bibinfo  {journal} {Arxiv preprint}\ } (\bibinfo {year} {2025})},\ \Eprint {https://arxiv.org/abs/2502.00998} {arXiv:2502.00998} \BibitemShut {NoStop}%
\bibitem [{\citenamefont {Bauer}\ and\ \citenamefont {de~la Fuente}(2025)}]{Bauer2025}%
  \BibitemOpen
  \bibfield  {author} {\bibinfo {author} {\bibfnamefont {A.}~\bibnamefont {Bauer}}\ and\ \bibinfo {author} {\bibfnamefont {J.~C.~M.}\ \bibnamefont {de~la Fuente}},\ }\bibfield  {title} {\bibinfo {title} {{Planar fault-tolerant circuits for non-Clifford gates on the 2D color code}},\ }\href {http://arxiv.org/abs/2505.05175} {\bibfield  {journal} {\bibinfo  {journal} {Arxiv preprint}\ } (\bibinfo {year} {2025})},\ \Eprint {https://arxiv.org/abs/2505.05175} {arXiv:2505.05175} \BibitemShut {NoStop}%
\bibitem [{\citenamefont {Kapustin}\ and\ \citenamefont {Thorngren}(2017)}]{kapustin2013higher}%
  \BibitemOpen
  \bibfield  {author} {\bibinfo {author} {\bibfnamefont {A.}~\bibnamefont {Kapustin}}\ and\ \bibinfo {author} {\bibfnamefont {R.}~\bibnamefont {Thorngren}},\ }\bibfield  {title} {\bibinfo {title} {{Higher symmetry and gapped phases of Gauge theories}},\ }\href {https://doi.org/10.1007/978-3-319-59939-7_5} {\bibfield  {journal} {\bibinfo  {journal} {Progress in Mathematics}\ }\textbf {\bibinfo {volume} {324}},\ \bibinfo {pages} {177} (\bibinfo {year} {2017})},\ \Eprint {https://arxiv.org/abs/1309.4721} {arXiv:1309.4721} \BibitemShut {NoStop}%
\bibitem [{\citenamefont {Gaiotto}\ \emph {et~al.}(2015)\citenamefont {Gaiotto}, \citenamefont {Kapustin}, \citenamefont {Seiberg},\ and\ \citenamefont {Willett}}]{Gaiotto2015}%
  \BibitemOpen
  \bibfield  {author} {\bibinfo {author} {\bibfnamefont {D.}~\bibnamefont {Gaiotto}}, \bibinfo {author} {\bibfnamefont {A.}~\bibnamefont {Kapustin}}, \bibinfo {author} {\bibfnamefont {N.}~\bibnamefont {Seiberg}},\ and\ \bibinfo {author} {\bibfnamefont {B.}~\bibnamefont {Willett}},\ }\bibfield  {title} {\bibinfo {title} {{Generalized global symmetries}},\ }\href {https://doi.org/10.1007/JHEP02(2015)172} {\bibfield  {journal} {\bibinfo  {journal} {Journal of High Energy Physics}\ }\textbf {\bibinfo {volume} {2015}},\ \bibinfo {pages} {1} (\bibinfo {year} {2015})},\ \Eprint {https://arxiv.org/abs/1412.5148} {arXiv:1412.5148} \BibitemShut {NoStop}%
\bibitem [{\citenamefont {Steane}(1996)}]{steane1996multiple}%
  \BibitemOpen
  \bibfield  {author} {\bibinfo {author} {\bibfnamefont {A.}~\bibnamefont {Steane}},\ }\bibfield  {title} {\bibinfo {title} {{Multiple-particle interference and quantum error correction}},\ }\href {https://royalsocietypublishing.org/doi/10.1098/rspa.1996.0136} {\bibfield  {journal} {\bibinfo  {journal} {Proceedings of the Royal Society of London. Series A: Mathematical, Physical and Engineering Sciences}\ }\textbf {\bibinfo {volume} {452}},\ \bibinfo {pages} {2551} (\bibinfo {year} {1996})}\BibitemShut {NoStop}%
\bibitem [{\citenamefont {Calderbank}\ and\ \citenamefont {Shor}(1996)}]{calderbank1996good}%
  \BibitemOpen
  \bibfield  {author} {\bibinfo {author} {\bibfnamefont {A.~R.}\ \bibnamefont {Calderbank}}\ and\ \bibinfo {author} {\bibfnamefont {P.~W.}\ \bibnamefont {Shor}},\ }\bibfield  {title} {\bibinfo {title} {{Good quantum error-correcting codes exist}},\ }\href {https://journals.aps.org/pra/abstract/10.1103/PhysRevA.54.1098} {\bibfield  {journal} {\bibinfo  {journal} {Physical Review A}\ }\textbf {\bibinfo {volume} {54}},\ \bibinfo {pages} {1098} (\bibinfo {year} {1996})}\BibitemShut {NoStop}%
\bibitem [{\citenamefont {Breuckmann}\ and\ \citenamefont {Eberhardt}(2021{\natexlab{b}})}]{breuckmann2021ldpc}%
  \BibitemOpen
  \bibfield  {author} {\bibinfo {author} {\bibfnamefont {N.~P.}\ \bibnamefont {Breuckmann}}\ and\ \bibinfo {author} {\bibfnamefont {J.~N.}\ \bibnamefont {Eberhardt}},\ }\bibfield  {title} {\bibinfo {title} {{Quantum Low-Density Parity-Check Codes}},\ }\href {https://doi.org/10.1103/PRXQuantum.2.040101} {\bibfield  {journal} {\bibinfo  {journal} {PRX Quantum}\ }\textbf {\bibinfo {volume} {2}},\ \bibinfo {pages} {040101} (\bibinfo {year} {2021}{\natexlab{b}})}\BibitemShut {NoStop}%
\bibitem [{\citenamefont {Rayhaun}\ and\ \citenamefont {Williamson}(2023)}]{Rayhaun2021}%
  \BibitemOpen
  \bibfield  {author} {\bibinfo {author} {\bibfnamefont {B.~C.}\ \bibnamefont {Rayhaun}}\ and\ \bibinfo {author} {\bibfnamefont {D.~J.}\ \bibnamefont {Williamson}},\ }\bibfield  {title} {\bibinfo {title} {{Higher-form subsystem symmetry breaking: Subdimensional criticality and fracton phase transitions}},\ }\bibfield  {journal} {\bibinfo  {journal} {SciPost Physics}\ }\textbf {\bibinfo {volume} {15}},\ \href {https://doi.org/10.21468/SciPostPhys.15.1.017} {10.21468/SciPostPhys.15.1.017} (\bibinfo {year} {2023}),\ \Eprint {https://arxiv.org/abs/2112.12735} {arXiv:2112.12735} \BibitemShut {NoStop}%
\bibitem [{\citenamefont {Williamson}(2016)}]{Williamson2016}%
  \BibitemOpen
  \bibfield  {author} {\bibinfo {author} {\bibfnamefont {D.~J.}\ \bibnamefont {Williamson}},\ }\bibfield  {title} {\bibinfo {title} {{Fractal symmetries: Ungauging the cubic code}},\ }\bibfield  {journal} {\bibinfo  {journal} {Physical Review B}\ }\textbf {\bibinfo {volume} {94}},\ \href {https://doi.org/10.1103/physrevb.94.155128} {10.1103/physrevb.94.155128} (\bibinfo {year} {2016}),\ \Eprint {https://arxiv.org/abs/1603.05182} {arXiv:1603.05182} \BibitemShut {NoStop}%
\bibitem [{\citenamefont {Kubica}\ and\ \citenamefont {Yoshida}(2018)}]{kubica2018ungauging}%
  \BibitemOpen
  \bibfield  {author} {\bibinfo {author} {\bibfnamefont {A.}~\bibnamefont {Kubica}}\ and\ \bibinfo {author} {\bibfnamefont {B.}~\bibnamefont {Yoshida}},\ }\bibfield  {title} {\bibinfo {title} {{Ungauging quantum error-correcting codes}},\ }\href {http://arxiv.org/abs/1805.01836} {\bibfield  {journal} {\bibinfo  {journal} {Arxiv preprint}\ ,\ \bibinfo {pages} {1}} (\bibinfo {year} {2018})},\ \Eprint {https://arxiv.org/abs/1805.01836} {arXiv:1805.01836} \BibitemShut {NoStop}%
\bibitem [{\citenamefont {Rakovszky}\ and\ \citenamefont {Khemani}(2023)}]{Rakovszky2023}%
  \BibitemOpen
  \bibfield  {author} {\bibinfo {author} {\bibfnamefont {T.}~\bibnamefont {Rakovszky}}\ and\ \bibinfo {author} {\bibfnamefont {V.}~\bibnamefont {Khemani}},\ }\bibfield  {title} {\bibinfo {title} {{The Physics of (good) LDPC Codes I. Gauging and dualities}},\ }\href {https://arxiv.org/abs/2310.16032v1} {\bibfield  {journal} {\bibinfo  {journal} {Arxiv preprint}\ } (\bibinfo {year} {2023})},\ \Eprint {https://arxiv.org/abs/2310.16032} {arXiv:2310.16032} \BibitemShut {NoStop}%
\bibitem [{\citenamefont {Zhang}\ and\ \citenamefont {Li}(2024)}]{zhang2024time}%
  \BibitemOpen
  \bibfield  {author} {\bibinfo {author} {\bibfnamefont {G.}~\bibnamefont {Zhang}}\ and\ \bibinfo {author} {\bibfnamefont {Y.}~\bibnamefont {Li}},\ }\bibfield  {title} {\bibinfo {title} {{Time-efficient logical operations on quantum {LDPC} codes}},\ }\href@noop {} {\bibfield  {journal} {\bibinfo  {journal} {arXiv preprint arXiv:2408.01339}\ } (\bibinfo {year} {2024})}\BibitemShut {NoStop}%
\bibitem [{\citenamefont {Hsieh}\ \emph {et~al.}(2025)\citenamefont {Hsieh}, \citenamefont {Li},\ and\ \citenamefont {Lin}}]{Hsieh2025Simplified}%
  \BibitemOpen
  \bibfield  {author} {\bibinfo {author} {\bibfnamefont {M.-H.}\ \bibnamefont {Hsieh}}, \bibinfo {author} {\bibfnamefont {X.}~\bibnamefont {Li}},\ and\ \bibinfo {author} {\bibfnamefont {T.-C.}\ \bibnamefont {Lin}},\ }\bibfield  {title} {\bibinfo {title} {{Simplified Quantum Weight Reduction with Optimal Bounds}},\ }\href {http://arxiv.org/abs/2510.09601} {\bibfield  {journal} {\bibinfo  {journal} {Arxiv preprint}\ } (\bibinfo {year} {2025})},\ \Eprint {https://arxiv.org/abs/2510.09601} {arXiv:2510.09601} \BibitemShut {NoStop}%
\bibitem [{\citenamefont {Haegeman}\ \emph {et~al.}(2015)\citenamefont {Haegeman}, \citenamefont {{Van Acoleyen}}, \citenamefont {Schuch}, \citenamefont {{Ignacio Cirac}},\ and\ \citenamefont {Verstraete}}]{Gaugingpaper}%
  \BibitemOpen
  \bibfield  {author} {\bibinfo {author} {\bibfnamefont {J.}~\bibnamefont {Haegeman}}, \bibinfo {author} {\bibfnamefont {K.}~\bibnamefont {{Van Acoleyen}}}, \bibinfo {author} {\bibfnamefont {N.}~\bibnamefont {Schuch}}, \bibinfo {author} {\bibfnamefont {J.}~\bibnamefont {{Ignacio Cirac}}},\ and\ \bibinfo {author} {\bibfnamefont {F.}~\bibnamefont {Verstraete}},\ }\bibfield  {title} {\bibinfo {title} {{Gauging quantum states: From global to local symmetries in many-body systems}},\ }\href {https://doi.org/10.1103/PhysRevX.5.011024} {\bibfield  {journal} {\bibinfo  {journal} {Physical Review X}\ }\textbf {\bibinfo {volume} {5}},\ \bibinfo {pages} {11024} (\bibinfo {year} {2015})},\ \Eprint {https://arxiv.org/abs/1407.1025} {arXiv:1407.1025} \BibitemShut {NoStop}%
\bibitem [{\citenamefont {Bulmash}\ and\ \citenamefont {Barkeshli}(2019)}]{Bulmash2019Gauging}%
  \BibitemOpen
  \bibfield  {author} {\bibinfo {author} {\bibfnamefont {D.}~\bibnamefont {Bulmash}}\ and\ \bibinfo {author} {\bibfnamefont {M.}~\bibnamefont {Barkeshli}},\ }\bibfield  {title} {\bibinfo {title} {{Gauging fractons: Immobile non-Abelian quasiparticles, fractals, and position-dependent degeneracies}},\ }\bibfield  {journal} {\bibinfo  {journal} {Physical Review B}\ }\textbf {\bibinfo {volume} {100}},\ \href {https://doi.org/10.1103/PhysRevB.100.155146} {10.1103/PhysRevB.100.155146} (\bibinfo {year} {2019}),\ \Eprint {https://arxiv.org/abs/1905.05771} {arXiv:1905.05771} \BibitemShut {NoStop}%
\bibitem [{\citenamefont {Prem}\ and\ \citenamefont {Williamson}(2019)}]{Prem2019Gauging}%
  \BibitemOpen
  \bibfield  {author} {\bibinfo {author} {\bibfnamefont {A.}~\bibnamefont {Prem}}\ and\ \bibinfo {author} {\bibfnamefont {D.}~\bibnamefont {Williamson}},\ }\bibfield  {title} {\bibinfo {title} {{Gauging permutation symmetries as a route to non-Abelian fractons}},\ }\href {https://doi.org/10.21468/scipostphys.7.5.068} {\bibfield  {journal} {\bibinfo  {journal} {SciPost Physics}\ }\textbf {\bibinfo {volume} {7}},\ \bibinfo {pages} {068} (\bibinfo {year} {2019})},\ \Eprint {https://arxiv.org/abs/1905.06309} {arXiv:1905.06309} \BibitemShut {NoStop}%
\bibitem [{\citenamefont {Yuan}\ \emph {et~al.}(2026)\citenamefont {Yuan}, \citenamefont {Lin}, \citenamefont {He}, \citenamefont {Cowtan},\ and\ \citenamefont {Williamson}}]{ParsimoniousSurgery}%
  \BibitemOpen
  \bibfield  {author} {\bibinfo {author} {\bibfnamefont {A.~C.}\ \bibnamefont {Yuan}}, \bibinfo {author} {\bibfnamefont {T.-C.}\ \bibnamefont {Lin}}, \bibinfo {author} {\bibfnamefont {Z.}~\bibnamefont {He}}, \bibinfo {author} {\bibfnamefont {A.}~\bibnamefont {Cowtan}},\ and\ \bibinfo {author} {\bibfnamefont {D.~J.}\ \bibnamefont {Williamson}},\ }\bibfield  {title} {\bibinfo {title} {Parsimonious quantum code surgery},\ }\href@noop {} {\bibfield  {journal} {\bibinfo  {journal} {\textit{in preparation}}\ } (\bibinfo {year} {2026})}\BibitemShut {NoStop}%
\bibitem [{\citenamefont {Shirley}\ \emph {et~al.}(2019)\citenamefont {Shirley}, \citenamefont {Slagle},\ and\ \citenamefont {Chen}}]{shirley2018FoliatedFracton}%
  \BibitemOpen
  \bibfield  {author} {\bibinfo {author} {\bibfnamefont {W.}~\bibnamefont {Shirley}}, \bibinfo {author} {\bibfnamefont {K.}~\bibnamefont {Slagle}},\ and\ \bibinfo {author} {\bibfnamefont {X.}~\bibnamefont {Chen}},\ }\bibfield  {title} {\bibinfo {title} {{Foliated fracton order from gauging subsystem symmetries}},\ }\bibfield  {journal} {\bibinfo  {journal} {SciPost Physics}\ }\textbf {\bibinfo {volume} {6}},\ \href {https://doi.org/10.21468/scipostphys.6.4.041} {10.21468/scipostphys.6.4.041} (\bibinfo {year} {2019}),\ \Eprint {https://arxiv.org/abs/1806.08679} {arXiv:1806.08679} \BibitemShut {NoStop}%
\bibitem [{\citenamefont {Campbell}(2019)}]{campbell2019theory}%
  \BibitemOpen
  \bibfield  {author} {\bibinfo {author} {\bibfnamefont {E.~T.}\ \bibnamefont {Campbell}},\ }\bibfield  {title} {\bibinfo {title} {{A theory of single-shot error correction for adversarial noise}},\ }\href {https://doi.org/10.1088/2058-9565/aafc8f} {\bibfield  {journal} {\bibinfo  {journal} {Quantum Science and Technology}\ }\textbf {\bibinfo {volume} {4}},\ \bibinfo {pages} {025006} (\bibinfo {year} {2019})}\BibitemShut {NoStop}%
\bibitem [{\citenamefont {Yoshida}(2015)}]{PhysRevB.91.245131}%
  \BibitemOpen
  \bibfield  {author} {\bibinfo {author} {\bibfnamefont {B.}~\bibnamefont {Yoshida}},\ }\bibfield  {title} {\bibinfo {title} {{Topological color code and symmetry-protected topological phases}},\ }\href {https://doi.org/10.1103/PhysRevB.91.245131} {\bibfield  {journal} {\bibinfo  {journal} {Physical Review B - Condensed Matter and Materials Physics}\ }\textbf {\bibinfo {volume} {91}},\ \bibinfo {pages} {245131} (\bibinfo {year} {2015})},\ \Eprint {https://arxiv.org/abs/1503.07208} {arXiv:1503.07208} \BibitemShut {NoStop}%
\bibitem [{\citenamefont {Else}\ and\ \citenamefont {Nayak}(2017)}]{else2017cheshire}%
  \BibitemOpen
  \bibfield  {author} {\bibinfo {author} {\bibfnamefont {D.~V.}\ \bibnamefont {Else}}\ and\ \bibinfo {author} {\bibfnamefont {C.}~\bibnamefont {Nayak}},\ }\bibfield  {title} {\bibinfo {title} {{Cheshire charge in (3+1)-dimensional topological phases}},\ }\bibfield  {journal} {\bibinfo  {journal} {Physical Review B}\ }\textbf {\bibinfo {volume} {96}},\ \href {https://doi.org/10.1103/PhysRevB.96.045136} {10.1103/PhysRevB.96.045136} (\bibinfo {year} {2017}),\ \Eprint {https://arxiv.org/abs/1702.02148} {arXiv:1702.02148} \BibitemShut {NoStop}%
\bibitem [{\citenamefont {Zhu}\ \emph {et~al.}(2023{\natexlab{b}})\citenamefont {Zhu}, \citenamefont {Sikander}, \citenamefont {Portnoy}, \citenamefont {Cross},\ and\ \citenamefont {Brown}}]{Zhu2023}%
  \BibitemOpen
  \bibfield  {author} {\bibinfo {author} {\bibfnamefont {G.}~\bibnamefont {Zhu}}, \bibinfo {author} {\bibfnamefont {S.}~\bibnamefont {Sikander}}, \bibinfo {author} {\bibfnamefont {E.}~\bibnamefont {Portnoy}}, \bibinfo {author} {\bibfnamefont {A.~W.}\ \bibnamefont {Cross}},\ and\ \bibinfo {author} {\bibfnamefont {B.~J.}\ \bibnamefont {Brown}},\ }\bibfield  {title} {\bibinfo {title} {{Non-Clifford and parallelizable fault-tolerant logical gates on constant and almost-constant rate homological quantum LDPC codes via higher symmetries}},\ }\href {http://arxiv.org/abs/2310.16982} {\bibfield  {journal} {\bibinfo  {journal} {Arxiv preprint}\ } (\bibinfo {year} {2023}{\natexlab{b}})},\ \Eprint {https://arxiv.org/abs/2310.16982} {arXiv:2310.16982} \BibitemShut {NoStop}%
\bibitem [{\citenamefont {Yoshida}(2016)}]{yoshida2015topological}%
  \BibitemOpen
  \bibfield  {author} {\bibinfo {author} {\bibfnamefont {B.}~\bibnamefont {Yoshida}},\ }\bibfield  {title} {\bibinfo {title} {{Topological phases with generalized global symmetries}},\ }\href {https://doi.org/10.1103/PhysRevB.93.155131} {\bibfield  {journal} {\bibinfo  {journal} {Physical Review B}\ }\textbf {\bibinfo {volume} {93}},\ \bibinfo {pages} {155131} (\bibinfo {year} {2016})},\ \Eprint {https://arxiv.org/abs/1508.03468} {arXiv:1508.03468} \BibitemShut {NoStop}%
\bibitem [{\citenamefont {Manjunath}\ \emph {et~al.}(2026)\citenamefont {Manjunath}, \citenamefont {Mattei}, \citenamefont {Tiwari},\ and\ \citenamefont {Ellison}}]{Ellison2025a}%
  \BibitemOpen
  \bibfield  {author} {\bibinfo {author} {\bibfnamefont {N.}~\bibnamefont {Manjunath}}, \bibinfo {author} {\bibfnamefont {V.}~\bibnamefont {Mattei}}, \bibinfo {author} {\bibfnamefont {A.}~\bibnamefont {Tiwari}},\ and\ \bibinfo {author} {\bibfnamefont {T.~D.}\ \bibnamefont {Ellison}},\ }\bibfield  {title} {\bibinfo {title} {Universal quantum computation with group surface codes},\ }\href@noop {} {\bibfield  {journal} {\bibinfo  {journal} {\textit{in preparation}}\ } (\bibinfo {year} {2026})}\BibitemShut {NoStop}%
\bibitem [{\citenamefont {Huang}\ \emph {et~al.}(2025)\citenamefont {Huang}, \citenamefont {Warman}, \citenamefont {Schafer-Nameki},\ and\ \citenamefont {Chen}}]{Huang2025a}%
  \BibitemOpen
  \bibfield  {author} {\bibinfo {author} {\bibfnamefont {S.-J.}\ \bibnamefont {Huang}}, \bibinfo {author} {\bibfnamefont {A.}~\bibnamefont {Warman}}, \bibinfo {author} {\bibfnamefont {S.}~\bibnamefont {Schafer-Nameki}},\ and\ \bibinfo {author} {\bibfnamefont {Y.}~\bibnamefont {Chen}},\ }\bibfield  {title} {\bibinfo {title} {{Hybrid Lattice Surgery: Non-Clifford Gates via Non-Abelian Surface Codes}},\ }\href {http://arxiv.org/abs/2510.20890} {\bibfield  {journal} {\bibinfo  {journal} {Arxiv preprint}\ } (\bibinfo {year} {2025})},\ \Eprint {https://arxiv.org/abs/2510.20890} {arXiv:2510.20890} \BibitemShut {NoStop}%
\bibitem [{\citenamefont {Warman}\ and\ \citenamefont {Schafer-Nameki}(2025)}]{Warman2025}%
  \BibitemOpen
  \bibfield  {author} {\bibinfo {author} {\bibfnamefont {A.}~\bibnamefont {Warman}}\ and\ \bibinfo {author} {\bibfnamefont {S.}~\bibnamefont {Schafer-Nameki}},\ }\bibfield  {title} {\bibinfo {title} {{Transversal Clifford-Hierarchy Gates via Non-Abelian Surface Codes}},\ }\href {http://arxiv.org/abs/2512.13777} {\bibfield  {journal} {\bibinfo  {journal} {Arxiv preprint}\ } (\bibinfo {year} {2025})},\ \Eprint {https://arxiv.org/abs/2512.13777} {arXiv:2512.13777} \BibitemShut {NoStop}%
\bibitem [{\citenamefont {Ren}\ \emph {et~al.}(2026)\citenamefont {Ren}, \citenamefont {Chang}, \citenamefont {Lyons}, \citenamefont {Putterman}, \citenamefont {Tantivasadakarn}, \citenamefont {Albert}, \citenamefont {Brown},\ and\ \citenamefont {Williamson}}]{GaugingLDPC}%
  \BibitemOpen
  \bibfield  {author} {\bibinfo {author} {\bibfnamefont {Y.}~\bibnamefont {Ren}}, \bibinfo {author} {\bibfnamefont {K.}~\bibnamefont {Chang}}, \bibinfo {author} {\bibfnamefont {A.}~\bibnamefont {Lyons}}, \bibinfo {author} {\bibfnamefont {H.}~\bibnamefont {Putterman}}, \bibinfo {author} {\bibfnamefont {N.}~\bibnamefont {Tantivasadakarn}}, \bibinfo {author} {\bibfnamefont {V.~V.}\ \bibnamefont {Albert}}, \bibinfo {author} {\bibfnamefont {B.~J.}\ \bibnamefont {Brown}},\ and\ \bibinfo {author} {\bibfnamefont {D.~J.}\ \bibnamefont {Williamson}},\ }\href@noop {} {\bibfield  {journal} {\bibinfo  {journal} {\textit{in preparation}}\ } (\bibinfo {year} {2026})}\BibitemShut {NoStop}%
\bibitem [{\citenamefont {Scruby}\ \emph {et~al.}(2022)\citenamefont {Scruby}, \citenamefont {Browne}, \citenamefont {Webster},\ and\ \citenamefont {Vasmer}}]{Scruby2022}%
  \BibitemOpen
  \bibfield  {author} {\bibinfo {author} {\bibfnamefont {T.~R.}\ \bibnamefont {Scruby}}, \bibinfo {author} {\bibfnamefont {D.~E.}\ \bibnamefont {Browne}}, \bibinfo {author} {\bibfnamefont {P.}~\bibnamefont {Webster}},\ and\ \bibinfo {author} {\bibfnamefont {M.}~\bibnamefont {Vasmer}},\ }\bibfield  {title} {\bibinfo {title} {{Numerical Implementation of Just-In-Time Decoding in Novel Lattice Slices through the Three-Dimensional Surface Code}},\ }\bibfield  {journal} {\bibinfo  {journal} {Quantum}\ }\textbf {\bibinfo {volume} {6}},\ \href {https://doi.org/10.22331/Q-2022-05-24-721} {10.22331/Q-2022-05-24-721} (\bibinfo {year} {2022}),\ \Eprint {https://arxiv.org/abs/2012.08536} {arXiv:2012.08536} \BibitemShut {NoStop}%
\bibitem [{\citenamefont {Scruby}\ \emph {et~al.}(2025)\citenamefont {Scruby}, \citenamefont {Nemoto},\ and\ \citenamefont {Cai}}]{Scruby2025}%
  \BibitemOpen
  \bibfield  {author} {\bibinfo {author} {\bibfnamefont {T.~R.}\ \bibnamefont {Scruby}}, \bibinfo {author} {\bibfnamefont {K.}~\bibnamefont {Nemoto}},\ and\ \bibinfo {author} {\bibfnamefont {Z.}~\bibnamefont {Cai}},\ }\bibfield  {title} {\bibinfo {title} {{Fault-tolerant quantum computation without distillation on a 2D device}},\ }\bibfield  {journal} {\bibinfo  {journal} {npj Quantum Information}\ }\textbf {\bibinfo {volume} {11}},\ \href {https://doi.org/10.1038/s41534-025-01133-7} {10.1038/s41534-025-01133-7} (\bibinfo {year} {2025}),\ \Eprint {https://arxiv.org/abs/2412.12529} {arXiv:2412.12529} \BibitemShut {NoStop}%
\bibitem [{\citenamefont {Zhu}\ \emph {et~al.}(2026)\citenamefont {Zhu}, \citenamefont {Kobayashi},\ and\ \citenamefont {Hsin}}]{Zhu2026NAQLDPC}%
  \BibitemOpen
  \bibfield  {author} {\bibinfo {author} {\bibfnamefont {G.}~\bibnamefont {Zhu}}, \bibinfo {author} {\bibfnamefont {R.}~\bibnamefont {Kobayashi}},\ and\ \bibinfo {author} {\bibfnamefont {P.-S.}\ \bibnamefont {Hsin}},\ }\bibfield  {title} {\bibinfo {title} {{Non-Abelian qLDPC: TQFT Formalism, Addressable Gauging Measurement and Application to Magic State Fountain on 2D Product Codes}},\ }\href {http://arxiv.org/abs/2601.06736} {\bibfield  {journal} {\bibinfo  {journal} {arXiv preprint}\ } (\bibinfo {year} {2026})},\ \Eprint {https://arxiv.org/abs/2601.06736} {arXiv:2601.06736} \BibitemShut {NoStop}%
\end{thebibliography}%

\end{document}